\documentclass[journal]{IEEEtran}

\usepackage{cite}
\usepackage{graphicx}
\usepackage[cmex10]{amsmath}
\usepackage{amsfonts}
\usepackage{amssymb}
\usepackage{subfigure}
\usepackage{float}
\usepackage{color}

\begin{document}
\newtheorem{proposition}{Proposition}
\newtheorem{theorem}{Theorem}
\newtheorem{corollary}{Corollary}
\newtheorem{lemma}{Lemma}
\newtheorem{problem}{Problem}
\newtheorem{remark}{Remark}
\newtheorem{assumption}{Assumption}
\newtheorem{claim}{Claim}
\newtheorem{algorithm}{Algorithm}
\newtheorem{definition}{Definition}

\title{Robust Transceiver Design for $K$-Pairs\\Quasi-Static MIMO Interference Channels via Semi-Definite Relaxation}

\author{Eddy~Chiu,~\IEEEmembership{Student~Member,~IEEE,}
        Vincent~K.~N.~Lau,~\IEEEmembership{Senior~Member,~IEEE},
        Huang~Huang,~\IEEEmembership{Student~Member,~IEEE,}
        Tao~Wu and Sheng~Liu
\thanks{Manuscript received December 28, 2009; revised May 28, 2010
and September 8, 2010; accepted September 9, 2010. The associate
editor coordinating the review of this paper and approving it for
publication was O.~Simeone. The paper was presented in part at the
Asilomar Conference on Signals, Systems, and Computing, Pacific
Grove, CA, November 2010.}
\thanks{E.~Chiu, V.~K.~N.~Lau, and H.~Huang are with the Department
of Electronic and Computer Engineering, Hong Kong University of
Science and Technology, Hong Kong (e-mail: echiua@ieee.org,
eeknlau@ust.hk, and huang@ust.hk).}
\thanks{T.~Wu and S.~Liu are with Huawei Technologies, Co. Ltd.,
China (e-mail: walnut@huawei.com and martin.liu@huawei.com).}}

\markboth{To appear in IEEE Transactions on Wireless
Communications,~2010} {{Chiu \MakeLowercase{\textit et. al.}}:
Robust Transceiver Design for $K$-Pairs Quasi-Static MIMO
Interference Channels via Semi-Definite Relaxation}

\maketitle

\begin{abstract}
In this paper, we propose a robust transceiver design for the
$K$-pair quasi-static MIMO interference channel. Each transmitter is
equipped with $M$ antennas, each receiver is equipped with $N$
antennas, and the $k^{\textrm{th}}$ transmitter sends $L_k$
independent data streams to the desired receiver. In the literature,
there exist a variety of theoretically promising transceiver designs
for the interference channel such as interference alignment-based
schemes, which have feasibility and practical limitations. In order
to address practical system issues and requirements, we consider a
transceiver design that enforces robustness against imperfect
channel state information (CSI) as well as fair performance among
the users in the interference channel. Specifically, we formulate
the transceiver design as an optimization problem to maximize the
worst-case signal-to-interference-plus-noise ratio among all users.
We devise a low complexity iterative algorithm based on alternative
optimization and semi-definite relaxation techniques. Numerical
results verify the advantages of incorporating into transceiver
design for the interference channel important practical issues such
as CSI uncertainty and fairness performance.
\end{abstract}



\section{Introduction} \label{Sec:Introduction}
In many wireless network scenarios, the channel is shared among
multiple systems. The coexisting systems create mutual interference,
which poses great challenges for communication systems design.
Conventionally, interference is either treated as noise in the weak
interference case \cite{Jnl:Interferce_as_noise:Tse} or canceled at
the receiver in the strong interference case
\cite{Jnl:Interferce_cancel:Carleial,
Bok:Fundamentals_wireless:Tse}. In the past decade, various schemes
are proposed to utilize multiple signaling dimensions for
interference avoidance and mitigation. In particular, in the recent
breakthrough work \cite{Jnl:IA:Cadambe_Jafar}, the authors show that
the paradigm of interference alignment (IA) can be exploited to
confine mutual interference to some lower dimensional subspace, so
that desired signals can be transmitted on interference-free
subspace. It is shown that this IA scheme, if feasible, is optimal
in the degree-of-freedom (DoF) sense. The results of
\cite{Jnl:IA:Cadambe_Jafar} has triggered a number of extensions
\cite{Cnf:IA_alternating_minimization:Peters_Heath,
Misc:Distributed_IA:Gomadam_Cadambe_Jafar} and related works
\cite{Misc:Real_IA_SISO:Motahari_Khandani,
Misc:Real_IA_MIMO:Ghasemi_Khandani}. These IA-based schemes, albeit
theoretically promising, have various limitations. First, IA-based
schemes require ideal conditions to be feasible such as perfect
channel state information (CSI) and very large dimensions on the
signal space. For example, the conventional IA scheme
\cite{Jnl:IA:Cadambe_Jafar} requires time or frequency extensions to
have feasible solutions. For $K$-pairs quasi-static MIMO
interference channels where time / frequency extensions are not
viable, the IA scheme \cite{Jnl:IA:Cadambe_Jafar} is only feasible
for $K \leq 3$ (cf. \cite{Cnf:IA_feasibility:Yetis_Jafar}). Second,
while IA-based schemes have promising DoF performance -- which is an
asymptotic performance measure for very high signal-to-noise ratio
(SNR) -- they are not optimal at medium SNR that correspond to
practical applications. When designing practical communication
systems for the interference channel, a number of technical issues
shall be considered. Specifically, in practice only imperfect CSI is
available and there are limited signaling dimensions. Moreover, it
is important to ensure satisfactory performance among all the
systems in the network.

In this paper, we consider the problem of robust transceiver design
for the $K$-pair quasi-static MIMO interference channel with
fairness considerations. Specifically, 1) we apply robust design
principles to provide resilience against CSI uncertainties; and 2)
we formulate the transceiver design as a precoder-decorrelator
optimization problem to maximize the worst-case
signal-to-interference-plus-noise ratio (SINR) among all users in
the interference channel. In the literature, precoder-decorrelator
optimization for worst-case SINR are proposed for broadcast and
point-to-point systems \cite{Jnl:RobustQosBroadcastMiso:Davidson,
Jnl:RobustQosP2PMimo:Palomar, Cnf:RobustQosBroadcastMimo:Boche,
Jnl:RobustQosP2PMimo:Miquel}. Specifically, in
\cite{Jnl:RobustQosBroadcastMiso:Davidson,
Jnl:RobustQosP2PMimo:Miquel} the authors consider precoding design
for the worst-case SINR in MISO broadcast channel, where it is shown
that the precoder optimization problem is always convex. In
\cite{Cnf:RobustQosBroadcastMimo:Boche} the authors consider
precoder-decorrelator design for the worst-case SINR MIMO broadcast
channel using an iterative algorithm based on solving convex
subproblems. On the other hand, in
\cite{Jnl:RobustQosP2PMimo:Palomar} the authors consider a
space-time coding scheme for the point-to-point channel with
imperfect channel knowledge. However, these existing works cannot be
extended to robust transceiver design for the MIMO interference
channel, which presents the following key technical challenges.

\textbf{The Precoder-Decorrelator Optimization Problem is NP-Hard:}
The precoder-decorrelator optimization problem for the interference
channel involves solving a separable homogeneous quadratically
constrained quadratic program (QCQP), which is NP-hard in general
\cite{Bok:Palomar, Bok:Palomar2}. One approach to facilitate solving
this class of problems is to apply semidefinite relaxation (SDR) by
relaxing rank constraints; this method was applied in precoding
design for MISO broadcast channel
\cite{Jnl:Rank_constrained_separable_SDP:Yongwei_Palomar,
Jnl:Downlink_beamforming:Ottersten} and for MISO multicast channel
\cite{Jnl:MBS-SDMA_using_SDR_with_perfect_CSI:Sidiropoulos_Luo,
Jnl:MBS_using_SDR_with_perfect_CSI:Sidiropoulos_Davidson_Luo}.
Although the resultant semidefinite program (SDP) may be solvable,
the optimization in general does not always have the desired rank
profile.

\textbf{Convergence of Alternative Optimization Algorithm:} Our
proposed solution is based on alternative optimization (AO). The
method of AO was proposed in \cite{Misc:Chan-Byoung_Chae,
Jnl:AO:Chan-Byoung_Chae} for precoder and decorrelator optimization
for multi-user MIMO broadcast channels. However, coupled with the
rank constrained SDP issues as well as the absence of
uplink-downlink duality (as in the case of broadcast channels)
\cite{Jnl:Beamforming_duality:Liu, Cnf:Beamforming_duality:Boche},
establishing the convergence proof of the AO algorithm in the
interference channel is non-trivial \cite{Jnl:AO_math_paper} and
traditional convergence proof \cite{Misc:Chan-Byoung_Chae,
Jnl:AO:Chan-Byoung_Chae} cannot be applied to our situations.

\emph{Notation}: In the sequel, we adopt the following notations.
$\mathbb{R}^{M \times N}$, $\mathbb{C}^{M \times N}$ and
$\mathbb{Z}^{M \times N}$ denote the set of real, complex and
integer $M \times N$ matrices, respectively; $\mathbb{R}_+$ denotes
the set of positive real numbers; upper and lower case letters
denote matrices and vectors, respectively; $\mathbb{H}^N$ denotes
the set of $N \times N$ Hermitian matrices; $\textbf{X} \succeq 0$
denotes that $\textbf{X}$ is a positive semi-definite matrix; $(
\cdot )^T$ and $( \cdot )^\dag$ denote transpose and Hermitian
transpose, respectively; $\textrm{rank} ( \cdot )$ and $\textrm{Tr}
( \cdot )$ denote matrix rank and trace, respectively;
$[\textbf{X}]_{(a,b)}$ denotes the $(a,b)^{\textrm{th}}$ element of
$\textbf{X}$; $|| \cdot ||$ denotes the Frobenius norm; $\mathcal{I}
( \cdot )$ denotes the indicate function; $\mathcal{K}$ denotes the
index set $\{ 1, \ldots, K \}$ and $\mathcal{L}_k$ denotes the index
set $\{ 1, \ldots, L_k \}$; $\textbf{0}_N$ denotes an $N \times 1$
vector of zeros and $\textbf{I}_{N}$ denotes an $N \times N$
identity matrix; $\mathbb{E}[ \cdot ]$ denotes expectation; and
$\mathcal{CN} ( \boldsymbol\mu, \boldsymbol \Phi )$ denotes complex
Gaussian distribution with mean $\boldsymbol \mu$ and covariance
matrix $\boldsymbol \Phi$.

\begin{figure}[t]
\centering
\includegraphics[width = 3.5in]{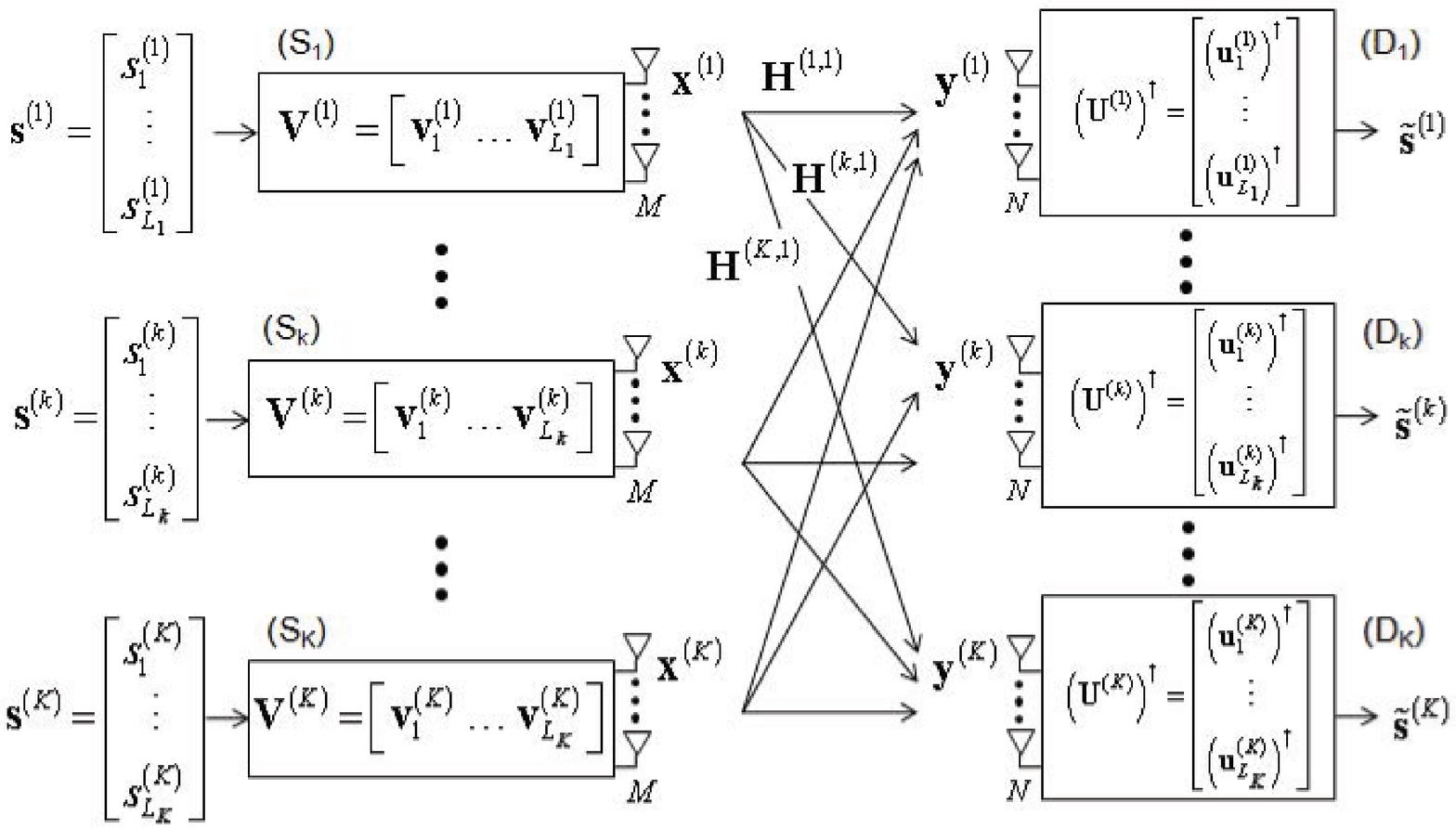}
\caption{System model. There are $K$ source-destination pairs where
each source node is equipped with $M$ antennas and each destination
node is equipped with $N$ antennas. The $k^{\textrm{th}}$
transmitter sends $L_k$ independent data streams to the desired
receiver.}\label{Fig:SystemModel}
\end{figure}

\section{System Model and Review of Prior Works} \label{Sec:System_model_prior_works}
\subsection{System Model} \label{Sec:System_Model}
We consider a MIMO interference channel consisting of $K$
source-destination pairs where each source node is equipped with $M$
antennas and each destination node is equipped with $N$ antennas as
shown in Fig.~\ref{Fig:SystemModel}. For ease of exposition, we
focus on the $k^\textrm{th}$ user referring to source node $S_k$ and
destination node $D_k$; nevertheless, the same model applies to all
other source-destination pairs. Specifically, $S_k$ transmits $L_k$
data streams $\textbf{s}^{(k)} = [ s_1^{(k)} \ldots s^{(k)}_{L_k}
]^T$ to $D_k$, which performs linear detection. The received signal
of $D_k$ is interfered by the transmitted signals of all other
users. To mitigate the impact of mutual interference, prior to
transmission $S_k$ precodes the data streams $\textbf{s}^{(k)}$
using the precoder matrix $\textbf{V}^{(k)} = [ \textbf{v}_1^{(k)}
\ldots \textbf{v}^{(k)}_{L_k} ] \in \mathbb{C}^{M \times L_k}$ and
$D_k$ decorrelates the received signal using the decorrelator matrix
$\textbf{U}^{(k)} = [ \textbf{u}_1^{(k)} \ldots
\textbf{u}^{(k)}_{L_k} ] \in \mathbb{C}^{N \times L_k}$. It follows
that the transmitted signal of $S_k$ is given by
\begin{equation}
\textbf{x}^{(k)} = \textbf{V}^{(k)} \textbf{s}^{(k)} =
\textstyle\sum_{l=1}^{L_k} \textbf{v}_l^{(k)}
s_l^{(k)},\label{Eqn:Tx_signal}
\end{equation}
the received signal of $D_k$ is given by
\begin{IEEEeqnarray*}{Rl}
\textbf{y}^{(k)}&= \textstyle\sum_{j=1}^K \textbf{H}^{(k,j)}
\textbf{x}^{(j)} + \textbf{n}^{(k)}\\
&= \textbf{H}^{(k,k)} \textbf{x}^{(k)} +
\underbrace{\textstyle\sum_{\substack{j=1\\j \neq k}}^K
\textbf{H}^{(k,j)} \textbf{x}^{(j)}}_{\textrm{interference}} +
\textbf{n}^{(k)},\IEEEyesnumber\label{Eqn:Rx_signal}
\end{IEEEeqnarray*}
and the decorrelator output of $D_k$ is given by
\begin{IEEEeqnarray*}{Rl}
\widetilde{\textbf{s}}^{(k)} &= ( \textbf{U}^{(k)} )^\dag
\textbf{y}^{(k)}\\
&= \underbrace{( \textbf{U}^{(k)} )^\dag \textbf{H}^{(k,k)}
\textbf{V}^{(k)} \textbf{s}^{(k)}}_{\textrm{desired signals}}
\IEEEyesnumber\label{Eqn:Decorrelator_output}\\
&+ \underbrace{\textstyle\sum_{\substack{j=1\\j \neq k}}^K (
\textbf{U}^{(k)} )^\dag \textbf{H}^{(k,j)} \textbf{V}^{(j)}
\textbf{s}^{(j)}}_{\textrm{leakage interference}} + (
\textbf{U}^{(k)} )^\dag \textbf{n}^{(k)},
\end{IEEEeqnarray*}
where $\textbf{H}^{(k,j)} \in \mathbb{C}^{N \times M}$ is the fading
channel from $S_j$ to $D_k$ and $\textbf{n}^{(k)} \sim \mathcal{CN}
( \textbf{0}_N, N_0 \textbf{I}_N )$ is the AWGN. As per
(\ref{Eqn:Tx_signal})--(\ref{Eqn:Decorrelator_output}), the estimate
of data stream $s_l^{(k)}$ is given by
\begin{IEEEeqnarray*}{Rl}
\widetilde{s}_l^{(k)} &= \underbrace{( \textbf{u}_l^{(k)} )^\dag
\textbf{H}^{(k,k)} \textbf{v}_l^{(k)} s_l^{(k)}}_{\textrm{desired
signal}} + \underbrace{\textstyle\sum_{\substack{m=1\\m \neq
l}}^{L_k} ( \textbf{u}_l^{(k)} )^\dag \textbf{H}^{(k,k)}
\textbf{v}_m^{(k)} s_m^{(k)}}_{\textrm{inter-stream
interference}}\\
&+ \underbrace{\textstyle\sum_{\substack{j=1\\j \neq k}}^K
\textstyle\sum_{m=1}^{L_j} ( \textbf{u}_l^{(k)} )^\dag
\textbf{H}^{(k,j)} \textbf{v}_m^{(j)} s_m^{(j)}}_{\textrm{leakage
interference}} + ( \textbf{u}_l^{(k)} )^\dag
\textbf{n}^{(k)},\;\;\;\;\IEEEyesnumber\label{Eqn:Signal_estimate_perfect_CSI}
\end{IEEEeqnarray*}
where the severity of the inter-stream and leakage interference
terms depend on the transceiver processing and CSI assumption.
Considering practical systems, we make the following assumptions
towards designing effective precoders and decorrelators.

\begin{assumption}[Transmit power constraint]\label{Assumption:Pwr}
We assume the data streams are independent and have unit power, i.e.
$\mathbb{E} [ ( \textbf{s}^{(k)} )^\dag \textbf{s}^{(k)} ] =
\textbf{I}_{L_k}$. Furthermore, we assume the maximum transmit power
of the $k^{\textrm{th}}$ source node is $P_k$ so the precoders shall
satisfy the power constraint $\mathbb{E} [ ( \textbf{x}^{(k)} )^\dag
\textbf{x}^{(k)} ] = \sum_{l=1}^{L_k} ( \textbf{v}_l^{(k)} )^\dag
\textbf{v}_l^{(k)} \leq P_k$.~ \hfill\IEEEQEDclosed
\end{assumption}

\begin{assumption}[Fading model]\label{Assumption:Fading}
We assume quasi-static fading so the fading channels
$\textbf{H}^{(k,j)}$ remain unchanged during a fading block. In
addition, we assume $\textrm{rank} ( \textbf{H}^{(k,j)} ) = \min (
M, N )$.~ \hfill\IEEEQEDclosed
\end{assumption}

\begin{assumption}[CSI model]\label{Assumption:CSI} We assume perfect
CSI is available at the receivers (i.e. perfect CSIR), and only
imperfect CSI is available at the transmitters (i.e. imperfect CSIT)
for designing the precoders and decorrelators. Specifically, we
model channel estimates at the transmitters as
\begin{equation}
\widehat{\textbf{H}}^{(k,j)} = \textbf{H}^{(k,j)} -
\mathbf{\Delta}^{(k,j)}, \forall j,k \in
\mathcal{K},\label{Eqn:Channel_estimates}
\end{equation}
where $\mathbf{\Delta}^{(k,j)}$ is the CSI error
\cite{Jnl:Worst-case_robust_MIMO:Jiaheng_Palomar,
Jnl:RobustQosBroadcastMiso:Davidson, Jnl:RobustQosP2PMimo:Palomar}.
Specifically, we assume $|| \mathbf{\Delta}^{(k,j)} ||^2 \leq
\varepsilon$, which implies that the actual channel
$\textbf{H}^{(k,j)}$ belongs to a spherical uncertainty region
centered at $\widehat{\textbf{H}}^{(k,j)}$ with radius
$\varepsilon$. For notational convenience, we denote $\mathcal{H} =
\{ \textbf{H}^{(k,j)} \}_{j,k=1}^K = \{
\widehat{\textbf{H}}^{(k,j)}\!\!+\!\mathbf{\Delta}^{(k,j)}
\}_{j,k=1}^K$ and $\mathcal{\widehat{H}} = \{
\widehat{\textbf{H}}^{(k,j)} \}_{j,k=1}^K$.~ \hfill\IEEEQEDclosed
\end{assumption}

\begin{remark}[Interpretation of the CSI error model]
The imperfect CSIT model (\ref{Eqn:Channel_estimates}) encapsulates
the following scenarios.
\begin{list}{\labelitemi}{\leftmargin=0.5em}
\item \emph{Quantized CSI in FDD Systems
\cite[Section~II\nobreakdash-B]{Jnl:RobustQosBroadcastMiso:Davidson}:}
For FDD systems, the transmitters are provided with quantized CSI
via feedback. Using uniform quantizers, the quantization cells in
the interior of the quantization region can be approximated by
spherical regions of radius equal to the quantization step size. As
a result, the imperfect CSIT model corresponds to quantized CSI
obtained using a uniform vector quantizer with quantization step
size $\sqrt{\varepsilon}$.

\item \emph{Estimated CSI in TDD Systems \cite[Section~IV\nobreakdash-A]{Jnl:RobustQosP2PMimo:Palomar}:}
For TDD systems, the transmitters can estimate the channels from the
sounding signals received in the reverse link. The imperfectness of
the CSIT in this case comes from the estimation noise as well as
delay. Using MMSE channel prediction, the CSI estimate
$\widehat{\textbf{H}}^{(k,j)}$ is unbiased, whereas the CSI error
$\mathbf{\Delta}^{(k,j)}$ is Gaussian distributed and independent
from the CSI estimate $\widehat{\textbf{H}}^{(k,j)}$. As a result,
$\mathbf{\Delta}^{(k,j)}$ is a jointly Gaussian matrix and $||
\mathbf{\Delta}^{(k,j)} ||^2 \leq \varepsilon$ corresponds to
``equal probability contour'' on the probability space of
$\mathbf{\Delta}^{(k,j)}$. In other words, the probability of the
event $|| \mathbf{\Delta}^{(k,j)} ||^2 \leq \varepsilon$ depends on
$\varepsilon$ only. Accordingly, we could find an $\varepsilon$ such
that $\textrm{Pr}[ || \mathbf{\Delta}^{(k,j)} ||^2 \leq \varepsilon
] = 0.99$ (for example).~ \hfill\IEEEQEDclosed
\end{list}
\end{remark}

By Assumptions~\ref{Assumption:Pwr}~to~\ref{Assumption:CSI}, the
data stream estimate $\widetilde{s}_l^{(k)}$ in
(\ref{Eqn:Signal_estimate_perfect_CSI}) can be \emph{equivalently}
expressed as
\begin{IEEEeqnarray*}{Rl}
\widetilde{s}_l^{(k)} &= ( \textbf{u}_l^{(k)} )^\dag (
\widehat{\textbf{H}}^{(k,k)}\!\!+\!\mathbf{\Delta}^{(k,k)} )
\textbf{v}_l^{(k)} s_l^{(k)}\\
+&\textstyle\sum_{\substack{m=1\\m \neq l}}^{L_k} (
\textbf{u}_l^{(k)} )^\dag (
\widehat{\textbf{H}}^{(k,k)}\!\!\!+\!\mathbf{\Delta}^{(k,k)} )
\textbf{v}_m^{(k)} s_m^{(k)}\IEEEyesnumber\label{Eqn:Signal_estimate_imperfect_CSI}\\
+&\textstyle\sum_{\substack{j=1\\j \neq k}}^K\!\sum_{m=1}^{L_j} (
\textbf{u}_l^{(k)} )^\dag (
\widehat{\textbf{H}}^{(k,j)}\!\!\!+\!\mathbf{\Delta}^{(k,j)} )
\textbf{v}_m^{(j)} s_m^{(j)}\!+\!( \textbf{u}_l^{(k)} )^\dag
\textbf{n}^{(k)}.
\end{IEEEeqnarray*}
The actual SINR of $\widetilde{s}_l^{(k)}$ at the $k^{\textrm{th}}$
receiver is given by (\ref{Eqn:Sinr_perfect}),
\begin{figure*}[!t]
\normalsize
\begin{IEEEeqnarray*}{l}
\gamma_l^{(k)} ( \mathcal{H}, \{ \{ \textbf{v}_m^{(j)}
\}_{m=1}^{L_j} \}_{j=1}^K, \textbf{u}_l^{(k)} ) \textstyle= \frac{||
( \textbf{u}_l^{(k)} )^\dag ( \widehat{\textbf{H}}^{(k,k)}\!+
\mathbf{\Delta}^{(k,k)} ) \textbf{v}_l^{(k)}
||^2}{\sum_{\substack{m=1\\m \neq l}}^{L_k} || ( \textbf{u}_l^{(k)}
)^\dag ( \widehat{\textbf{H}}^{(k,k)}\!+ \mathbf{\Delta}^{(k,k)} )
\textbf{v}_m^{(k)} ||^2 + \sum_{\substack{j=1\\j \neq k}}^K
\sum_{m=1}^{L_j} || ( \textbf{u}_l^{(k)} )^\dag (
\widehat{\textbf{H}}^{(k,j)}\!+ \mathbf{\Delta}^{(k,j)} )
\textbf{v}_m^{(j)} ||^2 + N_0 || \textbf{u}_l^{(k)}
||^2}\IEEEyesnumber\label{Eqn:Sinr_perfect}
\end{IEEEeqnarray*}
\hrulefill
\end{figure*}
whereby the instantaneous mutual information between data stream
$s_l^{(k)}$ and estimate $\widetilde{s}_l^{(k)}$ can be expressed as
\begin{IEEEeqnarray*}{l}
C_l^{(k)} ( \mathcal{H}, \{ \{ \textbf{v}_m^{(j)} \}_{m=1}^{L_j}
\}_{j=1}^K, \textbf{u}_l^{(k)} )\IEEEyesnumber\label{Eqn:Cap_perfect}\\
= \log_2( 1 + \gamma_l^{(k)} ( \mathcal{H}, \{ \{ \textbf{v}_m^{(j)}
\}_{m=1}^{L_j} \}_{j=1}^K, \textbf{u}_l^{(k)} ) ).
\end{IEEEeqnarray*}

\subsection{Review of Prominent Transceiver Designs for MIMO
Interference Channels}\label{Sec:Prior_works} In the following, we
review the motivations and issues of prominent transceiver designs
for MIMO interference channels in the literature.

\subsubsection{Interference Alignment in Quasi-Static MIMO Signal
Space} \label{Sec:Scheme_IA_MIMO} In \cite{Jnl:IA:Cadambe_Jafar,
Cnf:IA_feasibility:Yetis_Jafar} the authors exploited IA in
quasi-static MIMO signal space for precoder-decorrelator design.
Specifically, assuming perfect CSI, we could obtain precoders and
decorrelators that confine the interference on each destination node
to a lower dimension subspace, such that interference can be more
effectively removed. Note that IA is only feasible with sufficiently
large number of signaling dimensions. For the $K$-pair quasi-static
$N \times M$ MIMO interference channel, IA could achieve a DoF of $K
\frac{\min ( M, N )}{2}$ for $K \leq 3$ but might not be feasible
for $K > 3$. Moreover, IA is not optimal in general at medium SNR.
For example, consider the data stream estimate
$\widetilde{s}_l^{(k)}$ in
(\ref{Eqn:Signal_estimate_imperfect_CSI}); suppose IA is feasible
then
\begin{equation*}
( \textbf{u}_l^{(k)} )^\dag \widehat{\textbf{H}}^{(k,j)}
\textbf{v}_m^{(j)} = 0, j \neq k \textrm{ or } l \neq m,
\end{equation*}
and the actual SINR of the $l^{\textrm{th}}$ data stream at
$k^{\textrm{th}}$ receiver is given by
\begin{IEEEeqnarray*}{l}
\gamma_l^{(k)} ( \mathcal{H}, \{ \{ \textbf{v}_m^{(j)}
\}_{m=1}^{L_j} \}_{j=1}^K, \textbf{u}_l^{(k)} )\\
\textstyle= \frac{|| ( \textbf{u}_l^{(k)} )^\dag (
\widehat{\textbf{H}}^{(k,k)}\!+ \mathbf{\Delta}^{(k,k)} )
\textbf{v}_l^{(k)} ||^2}{\left(\substack{\sum_{\substack{m=1\\m \neq
l}}^{L_j} || ( \textbf{u}_l^{(k)} )^\dag \mathbf{\Delta}^{(k,k)}
\textbf{v}_m^{(k)} ||^2\\+ \sum_{\substack{j=1\\j \neq k}}^K
\sum_{m=1}^{L_j} || ( \textbf{u}_l^{(k)} )^\dag
\mathbf{\Delta}^{(k,j)} \textbf{v}_m^{(j)} ||^2 + N_0 ||
\textbf{u}_l^{(k)} ||^2}\right)}.\IEEEyesnumber\label{Eqn:Sinr_IA}
\end{IEEEeqnarray*}
As per (\ref{Eqn:Sinr_IA}), the presence of CSI error
$\mathbf{\Delta}^{(k,j)}$ creates persistent residual interference.
Even when the residual interference is negligible, i.e.
\begin{IEEEeqnarray*}{l}
\gamma_l^{(k)} ( \mathcal{H}, \{ \{ \textbf{v}_m^{(j)}
\}_{m=1}^{L_j} \}_{j=1}^K, \textbf{u}_l^{(k)} ) \textstyle\approx
\frac{|| ( \textbf{u}_l^{(k)} )^\dag (
\widehat{\textbf{H}}^{(k,k)}\!\!+\!\mathbf{\Delta}^{(k,k)} )
\textbf{v}_l^{(k)} ||^2}{N_0 || \textbf{u}_l^{(k)} ||^2},
\end{IEEEeqnarray*}
the conventional IA scheme \cite{Jnl:IA:Cadambe_Jafar,
Cnf:IA_feasibility:Yetis_Jafar} makes no attempt to optimize SINR
performance.

\subsubsection{Interference Alignment in Real Fading
Channels} \label{Sec:Scheme_IA_real} In
\cite{Misc:Real_IA_SISO:Motahari_Khandani,
Misc:Real_IA_MIMO:Ghasemi_Khandani} the authors consider IA along
the real line by creating fictitious signaling dimensions.
Specifically, assuming perfect CSI, we could design the leakage
interference terms at each destination node to have the same scaling
factor (or pseudo direction), such that interference can be
effectively removed. For example, consider the received signal in
(\ref{Eqn:Rx_signal}); for the purpose of illustration let $M=N=1$
and $H^{(k,j)} \in \mathbb{R}$ so
\begin{IEEEeqnarray*}{Rl}
y^{(k)} &\textstyle= H^{(k,k)} x^{(k)} + \sum_{\substack{j=1\\j \neq
k}}^K H^{(k,j)} x^{(j)} + n^{(k)}\\
&\textstyle\stackrel{(a)}{=} \sum_{l=1}^{L_j} (
\hat{H}^{(k,k)}\!+\!\Delta^{(k,k)} ) v_l^{(k)} s_l^{(k)}\\
&\textstyle+\sum_{\substack{j=1\\j \neq k}}^K \sum_{l=1}^{L_j}
\underbrace{( \hat{H}^{(k,j)}\!+\!\Delta^{(k,j)} ) v_l^{(j)}
s_l^{(j)}}_{\textrm{leakage interference}} +
n^{(k)},\label{Eqn:Rx_signal_real}
\end{IEEEeqnarray*}
where (a) follows from (\ref{Eqn:Tx_signal}) and
(\ref{Eqn:Channel_estimates}). To facilitate IA along the real line,
the data streams shall belong to the set of integers (i.e.
$s_l^{(k)} \in \mathbb{Z}$) and we shall choose the precoders such
that $\hat{H}^{(k,j)} v_l^{(j)} = \hat{H}^{(k,m)} v_l^{(m)}$ for $j
\neq m$. It is shown in \cite{Misc:Real_IA_SISO:Motahari_Khandani,
Misc:Real_IA_MIMO:Ghasemi_Khandani} that, if ideally CSI error is
negligible (i.e. $\Delta^{(k,j)} \approx 0$), this scheme could
theoretically achieve a DoF of $K \frac{MN}{M+N}$. However, this
scheme would require infinite SNR and cannot be implemented in
practice.

\subsubsection{Iterative Algorithms to Minimize Leakage
Interference / Maximize SINR} \label{Sec:Scheme_IA_greedy} In
\cite{Misc:Distributed_IA:Gomadam_Cadambe_Jafar,
Cnf:IA_alternating_minimization:Peters_Heath} the authors exploit
uplink-downlink duality and propose iterative algorithms for
precoder-decorrelator design. Specifically, the algorithms in
\cite[Algorithm~1]{Misc:Distributed_IA:Gomadam_Cadambe_Jafar},
\cite{Cnf:IA_alternating_minimization:Peters_Heath} are established
with the objective of sequentially minimizing the aggregate leakage
interference induced by each data stream, whereas the algorithm in
\cite[Algorithm~2]{Misc:Distributed_IA:Gomadam_Cadambe_Jafar} is
established with the objective to sequentially maximize the SINR of
each data stream. Note that the aforementioned algorithms neglect
the presence of CSI error, which could have significant performance
impacts. Moreover, these algorithms neglect individual user
performance and fairness. This is undesirable because for practical
systems it is important to ensure all users have satisfactory
performance.

\section{Problem Formulation: Robust Transceiver Design with
Fairness Considerations} \label{Sec:Problem_formulation} In this
section, we formulate a transceiver design for the $K$-pair
quasi-static MIMO interference channel that is robust against CSI
uncertainties and with the objective of enforcing fairness among all
users' data streams. Specifically, to provide the best resilience
against CSI error, we adopt a worst-case design approach. On the
other hand, the fairness aspect is motivated by the practical system
consideration to ensure all users in the network can have
satisfactory performance. As such, we formulate the
precoder-decorrelator design with imperfect CSIT as an optimization
problem to maximize the \emph{worst-case} SINR among all users' data
streams, subject to the maximum transmit power per source node.

\subsection{Optimization Problem}
The robust and fair transceiver optimization problem for the
$K$-pair $N \times M$ MIMO interference channel consists of the
following components.
\begin{list}{\labelitemi}{\leftmargin=0.5em}
\item \textbf{Optimization Variables}: The optimization variables
include the set of precoders $\{ \{ \textbf{v}_m^{(j)}
\}_{m=1}^{L_j} \}_{j=1}^K$ and the set of decorrelators $\{ \{
\textbf{u}_m^{(j)} \}_{m=1}^{L_j} \}_{j=1}^K$. These variables are
adaptive with respect to imperfect CSIT $\widehat{\mathcal{H}} = \{
\widehat{\textbf{H}}^{(k,j)} \}_{j,k=1}^K$.

\item \textbf{Optimization Objective}: The optimization objective is to
maximize, with imperfect CSIT, the minimum worst-case SINR
among\footnote{Note that (\ref{Eqn:Worst-case_SINR_Tx}) is the
worst-case SINR perceived by the transmitter based on imperfect CSIT
$\widehat{\mathcal{H}} = \{ \widehat{\textbf{H}}^{(k,j)}
\}_{j,k=1}^K$. We choose the worst-case SINR perceived by the
transmitter in order to incorporate robustness against CSI error
$\mathbf{\Delta}^{(k,j)}$.} all users' data streams (perceived by
the transmitter) given by (cf. (\ref{Eqn:Sinr_perfect}) and
Assumption~\ref{Assumption:CSI})
\end{list}
\begin{equation}
\min_{\substack{\forall k \in \mathcal{K}\\\forall l \in
\mathcal{L}_k}} \min_{|| \mathbf{\Delta}^{(k,j)} ||^2 \leq
\varepsilon} \gamma_l^{(k)} ( \mathcal{H}, \{ \{ \textbf{v}_m^{(j)}
\}_{m=1}^{L_j} \}_{j=1}^K, \textbf{u}_l^{(k)}
).\label{Eqn:Worst-case_SINR_Tx}
\end{equation}
\begin{list}{\labelitemi}{\leftmargin=0.5em}
\item \textbf{Optimization Constraints}: The optimization constraints
are the maximum transmit power for each source node $P_1, \ldots,
P_K$, which give the precoder power constraints $\sum_{l=1}^{L_k} (
\textbf{v}_l^{(k)} )^\dag \textbf{v}_l^{(k)} \leq P_k, \forall k \in
\mathcal{K}$ (cf. Assumption~\ref{Assumption:Pwr}).
\end{list}

Accordingly, the optimization problem can be formally written as
Problem~1.\\
\indent\emph{Problem 1: (Robust Max-Min Fair Precoder-Decorrelator
Design):}
\begin{IEEEeqnarray*}{cl}
\IEEEeqnarraymulticol{2}{l}{ \{ \{ \{ ( \textbf{v}_m^{(j)} )^\star
\}_{m=1}^{L_j} \}_{j=1}^K, \{ \{ ( \textbf{u}_m^{(j)} )^\star
\}_{m=1}^{L_j} \}_{j=1}^K \} = \mathcal{P}( P_1, \ldots, P_K )}\\
\arg\!\!\!\!\!\!\max_{\substack{\textbf{v}_m^{(j)} \in
\mathbb{C}^{M\!\times\!1}\\ \textbf{u}_m^{(j)} \in
\mathbb{C}^{N\!\times\!1}}} \!\min_{\substack{\forall k \in
\mathcal{K}\\\forall l \in \mathcal{L}_k}}
\!\min_{||\!\mathbf{\Delta}^{(k,j)}\!||^2 \leq
\varepsilon}&\gamma_l^{(k)}\!( \mathcal{H}, \{ \{ \textbf{v}_m^{(j)}
\}_{m=1}^{L_j} \}_{j=1}^K,
\textbf{u}_l^{(k)}\!)\IEEEyessubnumber\label{Eqn:Problem_P0_cost}\\
\textrm{s. t.}&\textstyle\sum_{l=1}^{L_k} ( \textbf{v}_l^{(k)}
)^\dag \textbf{v}_l^{(k)}\!\!\leq\!\!P_k, \forall
k\!\in\!\mathcal{K}.\IEEEyessubnumber\label{Eqn:Problem_P0_TxPwr_constraint}
\end{IEEEeqnarray*}

In (\ref{Eqn:Problem_P0_cost}), the worst-case SINR with imperfect
CSIT is given by the following proposition.
\begin{proposition}[Worst-Case SINR with Imperfect CSIT]
\label{Proposition:Worst-case_SINR} Given CSI estimates
$\widehat{\mathcal{H}} = \{ \widehat{\textbf{H}}^{(k,j)}
\}_{j,k=1}^K$ at the transmitter with error $||
\mathbf{\Delta}^{(k,j)} ||^2 \leq \varepsilon$, the worst-case SINR
of data stream estimate $\widetilde{s}_l^{(k)}$ perceived by the
transmitter can be expressed as (\ref{Eqn:Worst-case_SINR}).
\begin{figure*}[!t]
\normalsize
\begin{equation}
\begin{array}{l}
\displaystyle \min_{|| \mathbf{\Delta}^{(k,j)} ||^2 \leq
\varepsilon} \textstyle \gamma_l^{(k)} ( \mathcal{H}, \{ \{
\textbf{v}_m^{(j)} \}_{m=1}^{L_j} \}_{j=1}^K, \textbf{u}_l^{(k)}
)\\
= \frac{|| ( \textbf{u}_l^{(k)} )^\dag
\widehat{\textbf{H}}^{(k,k)} \textbf{v}_l^{(k)} ||^2 - \varepsilon
|| \textbf{u}_l^{(k)} ||^2 || \textbf{v}_l^{(k)} ||^2}{\sum_{j=1}^K
\sum_{m=1}^{L_j} || ( \textbf{u}_l^{(k)} )^\dag
\widehat{\textbf{H}}^{(k,j)} \textbf{v}_m^{(j)} ||^2 + \varepsilon
|| \textbf{u}_l^{(k)} ||^2 \sum_{j=1}^K \sum_{m=1}^{L_j} ||
\textbf{v}_m^{(j)} ||^2 - || ( \textbf{u}_l^{(k)} )^\dag
\widehat{\textbf{H}}^{(k,k)} \textbf{v}_l^{(k)} ||^2 - \varepsilon
|| \textbf{u}_l^{(k)} ||^2 ||
\textbf{v}_l^{(k)} ||^2 + N_0 || \textbf{u}_l^{(k)} ||^2}\\
\triangleq \widetilde{\gamma}_l^{(k)} ( \widehat{\mathcal{H}}, \{ \{
\textbf{v}_m^{(j)} \}_{m=1}^{L_j} \}_{j=1}^K, \textbf{u}_l^{(k)}
).\label{Eqn:Worst-case_SINR}
\end{array}
\end{equation}
\hrulefill
\end{figure*}
\end{proposition}
\begin{proof}
Please refer to Appendix~\ref{Proof:Worst-case_SINR}.
\end{proof}

Using Proposition~\ref{Proposition:Worst-case_SINR} and let
$\widetilde{P} = \min( P_1, \ldots, P_K )$ and $\rho_k = P_k /
\widetilde{P}$, we can recast Problem~$\mathcal{P}$ as
\begin{IEEEeqnarray*}{cl}
\IEEEeqnarraymulticol{2}{l}{ \{ \gamma^\star, \{ \{ (
\textbf{v}_m^{(j)} )^\star \}_{m=1}^{L_j} \}_{j=1}^K, \{ \{ (
\textbf{u}_m^{(j)} )^\star \}_{m=1}^{L_j}
\}_{j=1}^K \} = \mathcal{P}( \widetilde{P} )}\;\;\;\;\;\;\;\;\\
\min_{\substack{\textbf{v}_m^{(\!j\!)}\!\in \mathbb{C}^{\!M\!\times\!1}\\
\textbf{u}_m^{(\!j\!)}\!\in \mathbb{C}^{\!N\!\times\!1}\\\gamma \in
\mathbb{R}_+}}\!\!&\;\;-\gamma\IEEEyessubnumber\label{Eqn:Problem_P_cost}\\
\textrm{s. t.}&\widetilde{\gamma}_l^{(k)}\!(
\widehat{\mathcal{H}},\!\{ \{ \textbf{v}_m^{(j)} \}_{m=1}^{L_j}
\}_{j=1}^K,\!\textbf{u}_l^{(k)}\!)\!\ge\!\gamma,\!\forall
l\!\in\!\mathcal{L}_k,\!\forall
k\!\in\!\mathcal{K},\IEEEyessubnumber\label{Eqn:Problem_P_QoS_constraint}\\
&\textstyle\sum_{l=1}^{L_k} ( \textbf{v}_l^{(k)} )^\dag
\textbf{v}_l^{(k)} \leq \rho_k \widetilde{P},\forall k \in
\mathcal{K}.\IEEEyessubnumber\label{Eqn:Problem_P_TxPwr_constraint}
\end{IEEEeqnarray*}

\subsection{Properties of the Optimization Problem}
Note that it is not trivial to solve Problem~$\mathcal{P}$ since it
is non-convex and NP-hard in general as we elaborate below. In
Section~\ref{Sec:Iterative_soln}, we shall propose a low complexity
iterative algorithm for solving Problem~$\mathcal{P}$.

\subsubsection{Problem~$\mathcal{P}$ is a non-convex problem}
\label{Sec:Non-convex} The minimum SINR constraints in
(\ref{Eqn:Problem_P_QoS_constraint}) can be rearranged as
\begin{IEEEeqnarray*}{l}
\textstyle( 1 + \gamma ) || ( \textbf{u}_l^{(k)} )^\dag
\widehat{\textbf{H}}^{(k,k)} \textbf{v}_l^{(k)} ||^2 + ( \gamma - 1
) \varepsilon || \textbf{u}_l^{(k)} ||^2 || \textbf{v}_l^{(k)}
||^2\\
\textstyle- \gamma N_0 || \textbf{u}_l^{(k)} ||^2 - \gamma
\sum_{j=1}^K \sum_{m=1}^{L_j} || ( \textbf{u}_l^{(k)} )^\dag
\widehat{\textbf{H}}^{(k,j)} \textbf{v}_m^{(j)} ||^2\\
\textstyle- \gamma \varepsilon || \textbf{u}_l^{(k)} ||^2
\sum_{j=1}^K \sum_{m=1}^{L_j} || \textbf{v}_m^{(j)} ||^2 \geq
0,\label{Eqn:Problem_P_Qos_rearrange}
\end{IEEEeqnarray*}
which are non-convex inequalities consisting of non-positive linear
combinations of norms. Therefore, Problem~$\mathcal{P}$ is a
non-convex problem.

\begin{figure}[t]
\centering
\includegraphics[width = 3.5in]{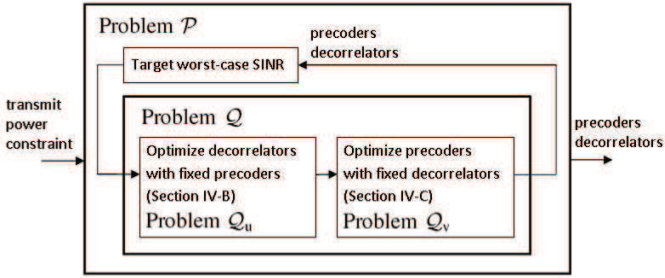}
\caption{Interrelationship among the optimization
problems.}\label{Fig:RelationshipAmongProblems}
\end{figure}

\subsubsection{Problem~$\mathcal{P}$ is NP-hard in general}
\label{Sec:NP-hard} To illustrate that Problem~$\mathcal{P}$ is
NP-hard in general, we consider the \emph{inverse problem} of
jointly minimizing the transmit powers of all source nodes subject
to a minimum SINR constraint for all users' data
streams\footnote{Please refer to
\cite{Jnl:Linear_precoding_conic:Eldar,
Jnl:Rank_constrained_separable_SDP:Yongwei_Palomar,
Jnl:Downlink_beamforming:Ottersten,
Jnl:MBS-SDMA_using_SDR_with_perfect_CSI:Sidiropoulos_Luo} and
references therein for discussions on the inverse relationship
between max-min fair and minimum power precoder design problems for
MISO \emph{broadcast} and \emph{multicast} channels.}. In
Section~\ref{Sec:Iterative_soln}, we shall propose an algorithm for
solving Problem~$\mathcal{P}$ \emph{facilitated} by solving the
inverse problem\footnote{The inverse problem will be utilized in
Section~\ref{Sec:Optimize_precoders}.} that consists of the
following components.
\begin{list}{\labelitemi}{\leftmargin=0.5em}
\item \textbf{Optimization Variables}: The optimization variables
include the set of precoders $\{ \{ \textbf{v}_m^{(j)}
\}_{m=1}^{L_j} \}_{j=1}^K$ and the set of decorrelators $\{ \{
\textbf{u}_m^{(j)} \}_{m=1}^{L_j} \}_{j=1}^K$.

\item \textbf{Optimization Objective}: The optimization
objective is to minimize the required transmit power of all source
nodes, by means of minimizing the precoder powers $\displaystyle
\textstyle\sum_{l=1}^{L_k} ( \textbf{v}_l^{(k)} )^\dag
\textbf{v}_l^{(k)}$, $\forall k \in \mathcal{K}$.

\item \textbf{Optimization Constraints}: The optimization constraint
is for all users' data streams to meet the prescribed minimum SINR
$\gamma$, i.e. $\widetilde{\gamma}_l^{(k)} ( \widehat{\mathcal{H}},
\{ \{ \textbf{v}_m^{(j)} \}_{m=1}^{L_j} \}_{j=1}^K,
\textbf{u}_l^{(k)} ) \ge \gamma$.
\end{list}

Accordingly, the inverse problem can be formally written as
Problem~2.\\
\indent\emph{Problem~2 (Power Minimization Precoder-Decorrelator
Design):}
\begin{IEEEeqnarray*}{cl}
\IEEEeqnarraymulticol{2}{l}{ \{ \beta^\star, \{ \{ (
\textbf{v}_m^{(j)} )^\star \}_{m=1}^{L_j} \}_{j=1}^K, \{ \{ (
\textbf{u}_m^{(j)} )^\star \}_{m=1}^{L_j}
\}_{j=1}^K \} = \mathcal{Q}( \gamma )}\;\;\;\;\;\;\;\;\\
\min_{\substack{\textbf{v}_m^{(\!j\!)}\!\in \mathbb{C}^{\!M\!\times\!1}\\
\textbf{u}_m^{(\!j\!)}\!\in \mathbb{C}^{\!N\!\times\!1}\\\beta \in
\mathbb{R}_+}}\!\!&\;\;\beta\IEEEyessubnumber\label{Eqn:Problem_Q_cost}\\
\textrm{s. t.}&\textstyle\sum_{l=1}^{L_k} ( \textbf{v}_l^{(k)}
)^\dag \textbf{v}_l^{(k)} \leq \rho_k \beta,\;\;\forall k \in
\mathcal{K},\IEEEyessubnumber\label{Eqn:Problem_Q_TxPwr_constraint}\\
&\widetilde{\gamma}_l^{(k)}\!( \widehat{\mathcal{H}},\!\{ \{
\textbf{v}_m^{(j)} \}_{m=1}^{L_j}
\}_{j=1}^K,\!\textbf{u}_l^{(k)}\!)\!\ge\!\gamma,\!\forall
l\!\in\!\mathcal{L}_k,\!\forall
k\!\in\!\mathcal{K}.\IEEEyessubnumber\label{Eqn:Problem_Q_QoS_constraint}
\end{IEEEeqnarray*}

Consider an instance of Problem~$\mathcal{Q}$ with minimum SINR
constraint $\widetilde{\gamma}$, i.e.
\begin{equation}
\{ \widetilde{\beta}, \{ \{ \widetilde{\textbf{v}}_m^{(j)}
\}_{m=1}^{L_j} \}_{j=1}^K, \{ \{ \widetilde{\textbf{u}}_m^{(j)}
\}_{m=1}^{L_j} \}_{j=1}^K \} = \mathcal{Q}( \widetilde{\gamma}
),\label{Eqn:Q_soln}
\end{equation}
and the required transmit power of the $k^{\textrm{th}}$ source node
is $\rho_k \widetilde{\beta}$. It can be shown that
\begin{equation}
\{ \widetilde{\gamma}, \{ \{ \widetilde{\textbf{v}}_m^{(j)}
\}_{m=1}^{L_j} \}_{j=1}^K, \{ \{ \widetilde{\textbf{u}}_m^{(j)}
\}_{m=1}^{L_j} \}_{j=1}^K \} = \mathcal{P}( \widetilde{\beta}
)\label{Eqn:P_Q_soln}
\end{equation}
so we can solve Problem~$\mathcal{Q}$ to obtain a corresponding
solution for Problem~$\mathcal{P}$, and vice-versa. Since
Problem~$\mathcal{Q}$ is NP-hard in general, Problem~$\mathcal{P}$
is also NP-hard. Specifically, we define the \emph{special case} of
Problem~$\mathcal{Q}$ with \emph{fixed} decorrelators as\\
\indent\emph{Problem~3 (Power Minimization Precoder Design
with Fixed Decorrelators):}
\begin{IEEEeqnarray*}{cl}
\IEEEeqnarraymulticol{2}{l}{ \{ \xi^\star, \{ \{ (
\textbf{v}_m^{(j)} )^\star \}_{m=1}^{L_j} \}_{j=1}^K \} =
\mathcal{Q}_\textrm{v}( \gamma, \{ \{
\textbf{u}_m^{(j)} \}_{m=1}^{L_j} \}_{j=1}^K )}\;\;\;\;\;\;\;\;\;\;\;\;\\
\min_{\substack{\textbf{v}_m^{(\!j\!)}\!\in \mathbb{C}^{\!M\!\times\!1}\\
\xi \in
\mathbb{R}_+}}\!\!&\;\;\xi\IEEEyessubnumber\label{Eqn:Problem_Q_V_cost}\\
\textrm{s. t.}&\textstyle\sum_{l=1}^{L_k} ( \textbf{v}_l^{(k)}
)^\dag \textbf{v}_l^{(k)} \leq \rho_k \xi, \forall k \in
\mathcal{K},\IEEEyessubnumber\label{Eqn:Problem_Q_V_TxPwr_constraint}\\
&\widetilde{\gamma}_l^{(k)}\!( \widehat{\mathcal{H}},\!\{ \{
\textbf{v}_m^{(j)} \}_{m=1}^{L_j}
\}_{j=1}^K,\!\textbf{u}_l^{(k)}\!)\!\ge\!\gamma,\!\forall
l\!\in\!\mathcal{L}_k,\!\forall
k\!\in\!\mathcal{K}.\IEEEyessubnumber\label{Eqn:Problem_Q_V_QoS_constraint}
\end{IEEEeqnarray*}
Note that Problem~$\mathcal{Q}_\textrm{v}$ belongs to the class of
separable homogenous QCQP, which is NP-hard in general
\cite{Bok:Palomar,Bok:Palomar2}. This implies that
Problem~$\mathcal{Q}$, which contains
Problem~$\mathcal{Q}_\textrm{v}$ as special case, is also NP-hard in
general\footnote{Problem~$\mathcal{Q}_\textrm{v}$ will be utilized
in Section~\ref{Sec:Optimize_precoders}.}.

\section{Low Complexity Iterative Solution} \label{Sec:Iterative_soln}
In this section, we propose a low complexity iterative algorithm for
solving the robust and fair transceiver optimization problem
$\mathcal{P}$. In particular, the proposed algorithm is facilitated
by solving the inverse Problem~$\mathcal{Q}$, whereby we exploit the
structure of Problem~$\mathcal{Q}$ to apply effective optimization
techniques.

\subsection{Overview of Algorithm} \label{Sec:Overview_algorithm}
The proposed algorithm for solving Problem~$\mathcal{P}$ is
facilitated by solving Problem~$\mathcal{Q}$ as illustrated in
Fig.~\ref{Fig:RelationshipAmongProblems}, which is also detailed in
Algorithm~\ref{Algorithm:Top-level}. Specifically, we iteratively
refine the decorrelators and precoders to monotonically improve the
minimum SINR. Each iteration consists of two stages:
\begin{list}{\labelitemi}{\leftmargin=0.5em}
\item (Steps 1-3 of Algorithm~\ref{Algorithm:Top-level}) First,
given the \emph{status quo} minimum SINR $\widetilde{\gamma}$
achieved with the $k^{\textrm{th}}$ source node transmitting at
power $P_k$, we solve Problem~$\mathcal{Q}$ to optimize the
precoders and decorrelators for minimizing the transmit powers, i.e.
\begin{equation}
\{ \widetilde{\beta}, \{ \{ \widetilde{\textbf{v}}_m^{(j)}
\}_{m=1}^{L_j} \}_{j=1}^K, \{ \{ \widetilde{\textbf{u}}_m^{(j)}
\}_{m=1}^{L_j} \}_{j=1}^K \} = \mathcal{Q}( \widetilde{\gamma}
),\label{Eqn:Target_SINR}
\end{equation}
such that the minimum SINR $\widetilde{\gamma}$ is achieved with the
$k^{\textrm{th}}$ source node transmitting at a \emph{reduced} power
of $\sum_{l=1}^{L_k} ( \widetilde{\textbf{v}}_l^{(k)} )^\dag
\widetilde{\textbf{v}}_l^{(k)} = \rho_k \widetilde{\beta} \leq P_k$.
\item (Steps 4-5 of Algorithm~\ref{Algorithm:Top-level}) Second,
we improve the minimum SINR by up-scaling the transmit precoding
power\footnote{We show in (\ref{Eqn:Convergence2}) that up-scaling
the precoding powers improve the minimum SINR.} of the
$k^{\textrm{th}}$ user to the power constraint $P_k$, i.e.
$\widetilde{\textbf{v}}_m^{(j)} = \sqrt{P_j / (\rho_j
\widetilde{\beta})} \widetilde{\textbf{v}}_m^{(j)}$.
\end{list}
We repeat the iteration until the minimum SINR converges to a
maximum. However, it is not trivial to solve the iteration step as
per (\ref{Eqn:Target_SINR}) since Problem~$\mathcal{Q}$ is NP-hard
in general as shown in Section~\ref{Sec:NP-hard}. As such, we shall
solve Problem~$\mathcal{Q}$ based on alternative optimization
between the decorrelators and the precoders, i.e. we present the
algorithm for optimizing the decorrelators with \emph{fixed}
precoders in Section~\ref{Sec:Optimize_decorrelators}, and introduce
the algorithm for optimizing the precoders with \emph{fixed}
decorrelators in Section~\ref{Sec:Optimize_precoders}. The top-level
detail steps of the optimization algorithm is summarized below
(Algorithm~\ref{Algorithm:Top-level}) and illustrated in
Fig~\ref{Fig:RelationshipAmongAlg}. The convergence proof for
Algorithm~\ref{Algorithm:Top-level} is provided in
Appendix~\ref{Proof:Convergence}.
\begin{algorithm}[Top-Level Algorithm]
\label{Algorithm:Top-level}~\\
\textbf{Inputs}: maximum transmit power for each
source node $P_1, \ldots, P_K$\\
\textbf{Outputs}: precoders $\{ \{ ( \textbf{v}_m^{(j)} )^\star
\}_{m=1}^{L_j} \}_{j=1}^K$ and decorrelators $\{ \{ (
\textbf{u}_m^{(j)} )^\star \}_{m=1}^{L_j} \}_{j=1}^K$
\begin{list}{\labelitemi}{\leftmargin=0.5em}
\item \textbf{Step 0}: Initialize decorrelators $\{ \{
\widetilde{\textbf{u}}_m^{(j)} \}_{m=1}^{L_j} \}_{j=1}^K$ and
precoders $\{ \{ \widetilde{\textbf{v}}_m^{(j)} \}_{m=1}^{L_j}
\}_{j=1}^K$, where the transmit power for the $j^{\textrm{th}}$
source node is $\sum_{m=1}^{L_j} ( \widetilde{\textbf{v}}_m^{(j)}
)^\dag \widetilde{\textbf{v}}_m^{(j)} = P_j$.
\end{list}
\textbf{Repeat}
\begin{list}{\labelitemi}{\leftmargin=0.5em}
\item \textbf{Step 1}: Optimize the decorrelators with fixed
precoders (cf. Section~\ref{Sec:Optimize_decorrelators})
\begin{equation*}
\{ \{ \widetilde{\textbf{u}}_m^{(j)} \}_{m=1}^{L_j} \}_{j=1}^K =
\mathcal{Q}_\textrm{u}( \{ \{ \widetilde{\textbf{v}}_m^{(j)}
\}_{m=1}^{L_j} \}_{j=1}^K ).
\end{equation*}
Update the candidate decorrelators $( \textbf{u}_m^{(j)} )^\star$ =
$\widetilde{\textbf{u}}_m^{(j)}$.
\item \textbf{Step 2}: Evaluate the minimum SINR
\begin{equation*}
\begin{array}{l}
\displaystyle\min_{\substack{\forall k \in \mathcal{K}\\\forall l
\in \mathcal{L}_k}} \widetilde{\gamma}_l^{(k)} (
\widehat{\mathcal{H}}, \{ \{ \widetilde{\textbf{v}}_m^{(j)}
\}_{m=1}^{L_j} \}_{j=1}^K, \widetilde{\textbf{u}}_l^{(k)} ) =
\widehat{\gamma}.
\end{array}
\end{equation*}
Update the target SINR $\widetilde{\gamma} = \widehat{\gamma}$.
\item \textbf{Step 3}: Optimize the precoders with fixed
decorrelators (cf. Section~\ref{Sec:Optimize_precoders})
\begin{center}
$\{ \xi, \{ \{ \widetilde{\textbf{v}}_m^{(j)} \}_{m=1}^{L_j}
\}_{j=1}^K \} = \mathcal{Q}_\textrm{v}( \widetilde{\gamma}, \{ \{
\widetilde{\textbf{u}}_m^{(j)} \}_{m=1}^{L_j} \}_{j=1}^K )$.
\end{center}
\item \textbf{Step 4}: Evaluate the required transmit power of each
source node $\rho_j \widetilde{\beta} = \sum_{m=1}^{L_j} (
\widetilde{\textbf{v}}_m^{(j)} )^\dag
\widetilde{\textbf{v}}_m^{(j)}$.
\item \textbf{Step 5}: Evaluate the minimum SINR with up-scaled
precoders
\begin{equation*}
\min_{\substack{\forall k \in \mathcal{K}\\\forall l \in
\mathcal{L}_k}} \widetilde{\gamma}_l^{(k)} ( \widehat{\mathcal{H}},
\{ \{ \sqrt{P_j / ( \rho_j \widetilde{\beta}} )
\widetilde{\textbf{v}}_m^{(j)} \}_{m=1}^{L_j} \}_{j=1}^K,
\widetilde{\textbf{u}}_l^{(k)} ) = \widehat{\gamma}.
\end{equation*}
Update the target SINR $\widetilde{\gamma} = \widehat{\gamma}$ and
candidate precoders $( \textbf{v}_m^{(j)} )^\star$ = $\sqrt{P_j / (
\rho_j \widetilde{\beta} )} \widetilde{\textbf{v}}_m^{(j)}$.
\end{list}
\textbf{Until} the minimum SINR $\widetilde{\gamma}$
converges.\\
\textbf{Return} precoders $\{ \{ ( \textbf{v}_m^{(j)} )^\star
\}_{m=1}^{L_j} \}_{j=1}^K$ and decorrelators $\{ \{ (
\textbf{u}_m^{(j)} )^\star \}_{m=1}^{L_j} \}_{j=1}^K$.
\end{algorithm}
\begin{figure}[t]
\centering
\includegraphics[width = 3.5in]{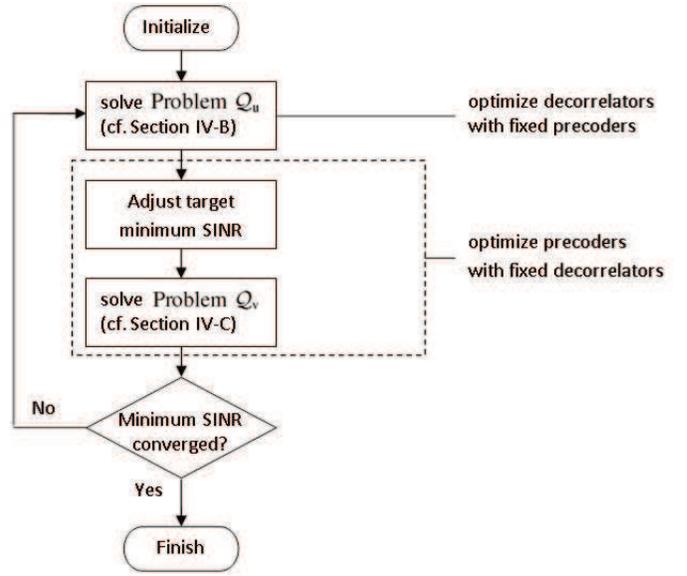}
\caption{Illustration of overall
algorithm.}\label{Fig:RelationshipAmongAlg}
\end{figure}

\subsection{Decorrelator Optimization with Fixed Precoders}
\label{Sec:Optimize_decorrelators} We define the decorrelator
optimization problem with fixed precoders to maximize the minimum
SINR among all users' data streams as\\
\indent\emph{Problem~4 (Maximum SINR Decorrelator Design with Fixed
Precoders):}
\begin{IEEEeqnarray*}{cl}
\IEEEeqnarraymulticol{2}{l}{ \{ \{ ( \textbf{u}_m^{(j)} )^\star
\}_{m=1}^{L_j} \}_{j=1}^K = \mathcal{Q}_\textrm{u}( \{ \{
\textbf{v}_m^{(j)} \}_{m=1}^{L_j}
\}_{j=1}^K )}\\
\arg\!\!\!\!\!\!\max_{\textbf{u}_m^{(j)} \in
\mathbb{C}^{N\!\times\!1}} \min_{\substack{\forall k \in
\mathcal{K}\\\forall l \in \mathcal{L}_k}}
\widetilde{\gamma}_l^{(k)} ( \widehat{\mathcal{H}}, \{ \{
\textbf{v}_m^{(j)} \}_{m=1}^{L_j} \}_{j=1}^K, \textbf{u}_l^{(k)}
).\IEEEyesnumber\label{Eqn:Problem_Q_U_QoS_constraint}
\end{IEEEeqnarray*}
As per (\ref{Eqn:Problem_Q_U_QoS_constraint}), the worst-case SINR
of data stream estimate $\widetilde{s}_l^{(k)}$ only depends on
decorrelator $\textbf{u}_l^{(k)}$. Therefore, we can independently
optimize each decorrelator, i.e.
\begin{IEEEeqnarray*}{cl}
\IEEEeqnarraymulticol{2}{l}{ ( \textbf{u}_l^{(k)} )^\star =
\mathcal{Q}_\textrm{u}^{(k,l)}( \{ \{
\textbf{v}_m^{(j)} \}_{l=1}^{L_j} \}_{j=1}^K )}\\
\arg\!\!\!\!\!\!\max_{\textbf{u}_l^{(k)} \in
\mathbb{C}^{N\!\times\!1}}&\;\;\widetilde{\gamma}_l^{(k)} (
\widehat{\mathcal{H}}, \{ \{ \textbf{v}_m^{(j)} \}_{m=1}^{L_j}
\}_{j=1}^K, \textbf{u}_l^{(k)}
),\IEEEyesnumber\label{Eqn:Problem_Q_U_cost}
\end{IEEEeqnarray*}
and the optimal decorrelator is given by
Theorem~\ref{Theorem:Optimal_decorrelator}.
\begin{theorem}[Optimal Decorrelator with Fixed Precoders]
\label{Theorem:Optimal_decorrelator} Given precoders $\{ \{
\textbf{v}_m^{(j)} \}_{m=1}^{L_j} \}_{j=1}^K$, the optimal
decorrelator for data stream estimate $\widetilde{s}_l^{(k)}$ is
given by $( \textbf{u}_l^{(k)} )^\star = \frac{( \textbf{F}_l^{(k)}
)^{-\frac{1}{2}} ( \textbf{w}_l^{(k)} )^\star}{|| (
\textbf{F}_l^{(k)} )^{-\frac{1}{2}} ( \textbf{w}_l^{(k)} )^\star
||}$, where
\begin{IEEEeqnarray*}{l}
\textbf{F}_l^{(k)} \textstyle= \sum_{j=1}^K \sum_{m=1}^{L_j}
\widehat{\textbf{H}}^{(k,j)} \textbf{v}_m^{(j)} ( \textbf{v}_m^{(j)}
)^\dag ( \widehat{\textbf{H}}^{(k,j)} )^\dag\\
\textstyle+ \varepsilon \sum_{j=1}^K \sum_{m=1}^{L_j} ||
\textbf{v}_m^{(j)} ||^2 \textbf{I}_N - \widehat{\textbf{H}}^{(k,k)}
\textbf{v}_l^{(k)} ( \textbf{v}_l^{(k)} )^\dag (
\widehat{\textbf{H}}^{(k,k)} )^\dag\\
- \varepsilon || \textbf{v}_l^{(k)} ||^2 \textbf{I}_N + N_0
\textbf{I}_N,
\end{IEEEeqnarray*}
$( \textbf{w}_l^{(k)} )^\star$ is the principle eigenvector of $(
\textbf{F}_l^{(k)} )^{-\frac{1}{2}} \textbf{E}_l^{(k)} (
\textbf{F}_l^{(k)} )^{-\frac{1}{2}}$, and $\textbf{E}_l^{(k)} =
\widehat{\textbf{H}}^{(k,k)} \textbf{v}_l^{(k)} ( \textbf{v}_l^{(k)}
)^\dag ( \widehat{\textbf{H}}^{(k,k)} )^\dag - \varepsilon ||
\textbf{v}_l^{(k)} ||^2 \textbf{I}_N$.
\end{theorem}
\begin{IEEEproof}
Please refer to Appendix~\ref{Proof:Optimal_decorrelator}.
\end{IEEEproof}

\subsection{Precoder Optimization with Fixed Decorrelators}
\label{Sec:Optimize_precoders} In Section~\ref{Sec:NP-hard}, we
defined the precoder optimization problem with fixed decorrelators,
Problem~$\mathcal{Q}_\textrm{v}$ (cf.
(\ref{Eqn:Problem_Q_V_cost})--(\ref{Eqn:Problem_Q_V_QoS_constraint})).
Since Problem~$\mathcal{Q}_\textrm{v}$ belongs to the class of
separable homogenous QCQP, it is NP-hard in general. In the
literature, some authors consider instances of this class of
problems for MISO \emph{broadcast} channel that are always solvable
(cf. \cite{Jnl:Linear_precoding_conic:Eldar,
Jnl:Downlink_beamforming:Ottersten} and references therein), whereas
some authors consider problems for MISO \emph{multicast} channel
that are always NP-hard (cf.
\cite{Jnl:MBS-SDMA_using_SDR_with_perfect_CSI:Sidiropoulos_Luo} and
references therein). For the \emph{interference} channel model
considered herein, we provide an algorithm for obtaining the optimal
solution for Problem~$\mathcal{Q}_\textrm{v}$.

One effective approach for solving separable homogenous QCQP is to
apply semidefinite relaxation (SDR) techniques. Let
$\textbf{V}_l^{(k)} = \textbf{v}_l^{(k)} ( \textbf{v}_l^{(k)}
)^\dag$. From (\ref{Eqn:Worst-case_SINR}), the worst-case SINR of
data stream estimate $s_l^{(k)}$ can be expressed as
(\ref{Eqn:Worst-case_SINR_V}).
\begin{figure*}[!t]
\normalsize
\begin{IEEEeqnarray*}{Rl}
\widetilde{\gamma}_l^{(k)} ( \widehat{\mathcal{H}}, \{ \{
\textbf{v}_m^{(j)} \}_{m=1}^{L_j} \}_{j=1}^K, \textbf{u}_l^{(k)} )
&\textstyle= \frac{\textrm{Tr} ( ( \widehat{\textbf{H}}^{(k,k)}
)^\dag \textbf{u}_l^{(k)} ( \textbf{u}_l^{(k)} )^\dag
\widehat{\textbf{H}}^{(k,k)} \textbf{V}_l^{(k)} ) - \varepsilon ||
\textbf{u}_l^{(k)} ||^2 \textrm{Tr} ( \textbf{V}_l^{(k)} )}{\left(
\substack{\sum_{j=1}^K \sum_{m=1}^{L_j} \textrm{Tr} ( (
\widehat{\textbf{H}}^{(k,j)} )^\dag \textbf{u}_l^{(k)} (
\textbf{u}_l^{(k)} )^\dag \widehat{\textbf{H}}^{(k,j)}
\textbf{V}_m^{(j)} ) + \varepsilon || \textbf{u}_l^{(k)} ||^2
\sum_{j=1}^K \sum_{m=1}^{L_j} \textrm{Tr} ( \textbf{V}_m^{(j)} )\\-
\textrm{Tr} ( ( \widehat{\textbf{H}}^{(k,k)} )^\dag
\textbf{u}_l^{(k)} ( \textbf{u}_l^{(k)} )^\dag
\widehat{\textbf{H}}^{(k,k)} \textbf{V}_l^{(k)} ) - \varepsilon ||
\textbf{u}_l^{(k)} ||^2 \textrm{Tr} ( \textbf{V}_l^{(k)} ) + N_0 ||
\textbf{u}_l^{(k)} ||^2}
\right)}\;\;\;\;\;\;\IEEEyesnumber\label{Eqn:Worst-case_SINR_V}\\
&\textstyle\triangleq \widetilde{\Gamma}_l^{(k)} (
\widehat{\mathcal{H}}, \{ \{ \textbf{V}_m^{(j)} \}_{m=1}^{L_j}
\}_{j=1}^K, \textbf{u}_l^{(k)} ).
\end{IEEEeqnarray*}
\hrulefill
\end{figure*}
It follows that we can \emph{equivalently} express the precoder
optimization problem with fixed decorrelators as
\begin{IEEEeqnarray*}{cl}
\IEEEeqnarraymulticol{2}{l}{ \{ \Xi^\star, \{ \{ (
\textbf{V}_m^{(j)} )^\star \}_{m=1}^{L_j} \}_{j=1}^K \} =
\widetilde{\mathcal{Q}}_\textrm{v}( \gamma, \{ \{
\textbf{u}_m^{(j)} \}_{m=1}^{L_j} \}_{j=1}^K )}\;\;\;\;\;\;\;\;\;\;\;\;\\
\min_{\substack{\textbf{V}_m^{(\!j\!)}\!\in
\mathbb{C}^{\!M\!\times\!M}\\\Xi \in
\mathbb{R}_+}}\!\!&\;\;\Xi\IEEEyessubnumber\label{Eqn:Problem_R_V_cost}\\
\textrm{s. t.}&\textstyle\sum_{l=1}^{L_k} \textrm{Tr} (
\textbf{V}_l^{(k)} ) \leq \rho_k \Xi, \forall k \in
\mathcal{K},\IEEEyessubnumber\label{Eqn:Problem_R_V_TxPwr_constraint}\\
&\widetilde{\Gamma}_l^{(k)}\!( \widehat{\mathcal{H}},\!\{ \{
\textbf{V}_m^{(j)} \}_{m=1}^{L_j}
\}_{j=1}^K,\!\textbf{u}_l^{(k)}\!)\!\ge\!\gamma,\!\forall
l\!\in\!\mathcal{L}_k,\!\forall
k\!\in\!\mathcal{K},\IEEEyessubnumber\label{Eqn:Problem_R_V_QoS_constraint}\\
&\textbf{V}_l^{(k)} \succeq 0, \forall k \in \mathcal{K}, \forall l
\in
\mathcal{L}_k,\IEEEyessubnumber\label{Eqn:Problem_R_V_SD_constraint}\\
&\textrm{rank} ( \textbf{V}_l^{(k)} ) = 1, \forall k \in
\mathcal{K}, \forall l \in
\mathcal{L}_k,\IEEEyessubnumber\label{Eqn:Problem_R_V_rank_constraint}
\end{IEEEeqnarray*}
where (\ref{Eqn:Problem_R_V_SD_constraint}) and
(\ref{Eqn:Problem_R_V_rank_constraint}) follow from the definition
of $\textbf{V}_l^{(k)}$, (\ref{Eqn:Problem_R_V_TxPwr_constraint})
are power constraints, and (\ref{Eqn:Problem_R_V_QoS_constraint})
are SINR constraints. Note that we could obtain the optimal precoder
$( \textbf{v}_m^{(j)} )^\star$ from the eigenvector of $(
\textbf{V}_m^{(j)} )^\star$ corresponding to the only non-zero
eigenvalue.

Comparing between Problem~$\widetilde{\mathcal{Q}}_\textrm{v}$ and
Problem~$\mathcal{Q}_\textrm{v}$, the SINR constraints of
Problem~$\widetilde{\mathcal{Q}}_\textrm{v}$
(\ref{Eqn:Problem_R_V_QoS_constraint}) are convex inequalities, i.e.
\begin{IEEEeqnarray*}{l}
( 1\!+\!\gamma ) \textrm{Tr} ( ( \widehat{\textbf{H}}^{(k,k)} )^\dag
\textbf{u}_l^{(k)} ( \textbf{u}_l^{(k)} )^\dag
\widehat{\textbf{H}}^{(k,k)} \textbf{V}_l^{(k)} ) - \gamma N_0 ||
\textbf{u}_l^{(k)} ||^2\\
\textstyle+ ( \gamma\!-\!1 )\varepsilon || \textbf{u}_l^{(k)} ||^2
\textrm{Tr} ( \textbf{V}_l^{(k)} ) - \gamma \varepsilon ||
\textbf{u}_l^{(k)} ||^2 \sum_{j=1}^K \sum_{m=1}^{L_j} \textrm{Tr} ( \textbf{V}_m^{(j)} )\\
\textstyle- \gamma \sum_{j=1}^K \sum_{m=1}^{L_j} \textrm{Tr} ( (
\widehat{\textbf{H}}^{(k,j)} )^\dag \textbf{u}_l^{(k)} (
\textbf{u}_l^{(k)} )^\dag \widehat{\textbf{H}}^{(k,j)}
\textbf{V}_m^{(j)} ) \ge 0,
\end{IEEEeqnarray*}
but Problem~$\widetilde{\mathcal{Q}}_\textrm{v}$ is still a
non-convex problem due to the rank constraints
(\ref{Eqn:Problem_R_V_rank_constraint}). By means of SDR, we
\emph{neglect} the rank constraints and
Problem~$\widetilde{\mathcal{Q}}_\textrm{v}$ degenerates into an SDP
that can be solved efficiently \cite{Bok:Convex_optimization:Boyd}.
In general, the resultant solution $\{ \Xi^\star, \{ \{ (
\textbf{V}_m^{(j)} )^\star \}_{m=1}^{L_j} \}_{j=1}^K \}$ could have
arbitrary rank. If $\textrm{rank} ( ( \textbf{V}_m^{(j)} )^\star ) =
1$, $\forall m \in \mathcal{L}_j$ and $\forall j \in \mathcal{K}$,
then constraints (\ref{Eqn:Problem_R_V_rank_constraint}) are
intrinsically satisfied and $\{ \{ ( \textbf{V}_m^{(j)} )^\star
\}_{m=1}^{L_j} \}_{j=1}^K$ are optimal. The following theorem
summarizes the optimality of the SDR solution in
(\ref{Eqn:Problem_R_V_cost})--(\ref{Eqn:Problem_R_V_rank_constraint}).

\begin{theorem}[Optimality of the SDR Solution]\label{Theorem:Optimiality_conditions}
The SDR solution of Problem~$\widetilde{\mathcal{Q}}_{\textrm{v}}$
will always give rank 1 solutions (i.e. $\textrm{rank} ( (
\textbf{V}_m^{(j)} )^\star ) = 1$) and hence, the SDR solution is
optimal for $\widetilde{\mathcal{Q}}_{\textrm{v}}$.
\end{theorem}
\begin{IEEEproof}
Please refer to Appendix~\ref{Proof:Optimiality_conditions}.
\end{IEEEproof}

\section{Simulation Results and Discussions} \label{Sec:Simulation_results}
In this section, we evaluate the proposed robust transceiver design
via numerical simulations. In particular, we compare the performance
of the proposed scheme against four baseline schemes:
\begin{list}{\labelitemi}{\leftmargin=0em}
\item \textbf{Baseline 1} the conventional IA scheme
\cite{Jnl:IA:Cadambe_Jafar};
\item \textbf{Baseline 2} the SINR maximization scheme
\cite[Algorithm~2]{Misc:Distributed_IA:Gomadam_Cadambe_Jafar};
\item \textbf{Baseline 3} a naive max-min SINR scheme adopted from
\cite{Cnf:RobustQosBroadcastMimo:Boche};
\item \textbf{Baseline 4} a naive max-min SINR scheme adopted from
\cite{Jnl:Linear_precoding_conic:Eldar}.
\end{list}
As discussed in Section~\ref{Sec:Prior_works}, baselines 1 and 2 are
theoretically promising schemes for the interference channel but
neglect important practical issues such as CSI uncertainty and
fairness among users. On the other hand, baselines 3 and 4 are
adopted from existing max-min SINR schemes that are originally
designed for the broadcast channel (i.e. there is \emph{only a
single transmitter} and multiple receivers). Without loss of
generality, we assume independent and identically distributed (iid)
Rayleigh fading channels, i.e. $[\textbf{H}^{(k,j)}]_{(a,b)} \sim
\mathcal{CN} ( 0, 1 )$, $\forall j,k \in \mathcal{K}$, $\forall a
\in [1,N]$, and $\forall b \in [1,M]$. For the purpose of
illustration, we consider the scenario where all users have the same
power constraint $P_1 = \ldots = P_K = P$. In
Fig.~\ref{Fig:K3_M4_L2_WorstStream} to Fig.~\ref{Fig:Imperfect} we
present simulation results for the average data rates\footnote{The
average data rate is defined as the average goodput (i.e. the
bits/s/Hz successfully delivered to the receiver). Specifically, the
goodput of data stream $s^{(k)}_{l}$ is given by $r_l^{(k)}
\mathcal{I} ( r_l^{(k)} \leq C_l^{(k)} )$, where $r_l^{(k)} =
\log_2( 1 + \gamma_l^{(k)}( \widehat{\mathcal{H}}, \{ \{
\textbf{v}_m^{(j)} \}_{m=1}^{L_j} \}_{j=1}^K, \textbf{u}_l^{(k)} )
)$ is the scheduled data rate based on the SINR perceived with
respect to imperfect CSIT $\widehat{\mathcal{H}} = \{
\widehat{\textbf{H}}^{(k,j)} \}_{j,k=1}^K$, and $C_l^{(k)} = \log_2(
1 + \gamma_l^{(k)}( \mathcal{H}, \{ \{ \textbf{v}_m^{(j)}
\}_{m=1}^{L_j} \}_{j=1}^K, \textbf{u}_l^{(k)} ) )$ is the actual
instantaneous mutual information.} versus SNR\footnote{The SNR is
defined as $\frac{P}{N_0}$, where $N_0$ is the AWGN variance.} with
different number of users and levels of CSI uncertainty.

\subsection{Fairness Performance}
In Fig.~\ref{Fig:K3_M4_L2_WorstStream} and
Fig.~\ref{Fig:K3_M4_L2_Sum}, we compare the average data rates of
the proposed and baseline schemes. For the purpose of illustration,
we consider the three-user $4 \times 4$ MIMO interference channel,
where each user transmits $L=2$ data streams and the precoders are
designed with imperfect CSIT with $\varepsilon = \{ 0.1, 0.15 \}$,
whereas the receivers have perfect CSIR. It can be observed that the
proposed scheme achieves much higher average worst-case data rate
per user than all the baseline schemes, and thus provides better
minimum performance. For example, at CSI error $\varepsilon = 0.15$,
the proposed scheme has 5dB SNR gain over the SINR maximization
algorithm (baseline 2) at providing a worst case data rate of 6
b/s/Hz and the conventional IA scheme (baseline 1) cannot provide
worst-case data rate of 6 b/s/Hz. The superior performances of the
proposed scheme is accountable to both the SDR approach as well as a
suitably chosen utility function (optimizing the worst case
performance). Specifically, the chosen utility function 1) provide
resilience against CSI uncertainties as well as 2) achieve fair
performance among users. On the other hand, the SDR approach also
contributes to obtaining a good solution for solving the
optimization problem.

\subsection{Total Sum Data Rate Performance}
In Fig.~\ref{Fig:K3_M4_L2_Sum}, we compare the average total sum
data rates of the proposed and baseline schemes for $K=3$, $N=M=4$,
$L=2$, and CSI error $\varepsilon = \{ 0.1, 0.15 \}$. It can be
observed that the proposed scheme not only achieves better
worst-case data rate but also achieves higher total sum data rate
than all the baseline schemes. In particular, due to the presence of
CSI error, the total sum rate of the conventional IA scheme
(baseline 1) does not scale linearly with the SNR anymore. Comparing
Fig.~\ref{Fig:K3_M4_L2_Sum} with
Fig.~\ref{Fig:K3_M4_L2_WorstStream}, it can be observed that the
proposed scheme achieves the performance gain on fairness without
sacrificing the total sum data rate.

\begin{figure}[t]
\centering
\includegraphics[width = 3.5in]{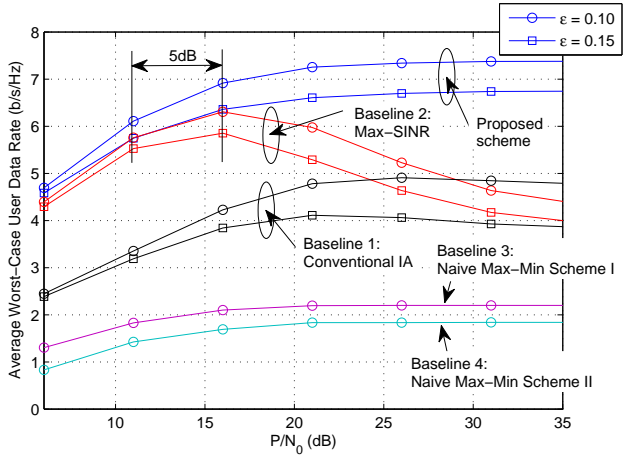}
\caption{Average data rate of the worst-case user versus SNR. $K=3$,
$N=M=4$, $L=2$ and CSI error $\varepsilon = \{0.1 ,
0.15\}$.}\label{Fig:K3_M4_L2_WorstStream}
\end{figure}

\begin{figure}[t]
\centering
\includegraphics[width = 3.5in]{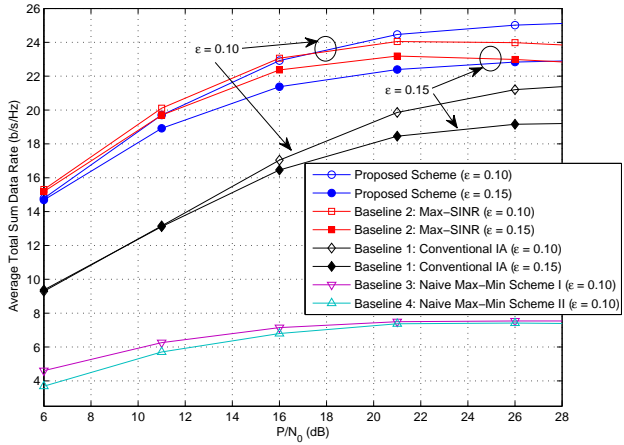}
\caption{Average total sum data rate versus SNR. $K=3$, $N=M=4$,
$L=2$ and CSI error $\varepsilon = \{0.1 ,
0.15\}$.}\label{Fig:K3_M4_L2_Sum}
\end{figure}

\subsection{Robustness to CSI Errors}
In Fig.~\ref{Fig:Imperfect}, we show the average worst-case data
rates of the proposed and baseline schemes for different levels of
CSI uncertainty. It can be observed that the proposed scheme always
achieves higher average worst-case data rate than the baseline
schemes. For example, the SINR maximization algorithm (baseline 2)
is designed assuming perfect CSI; its performance degrades rapidly
for CSI error $\varepsilon > 0.02$ and it can be observed that the
achieved data rate could decrease with increasing SNR. On the other
hand, the proposed scheme achieves a robust degradation with respect
to CSI errors.

\begin{figure}[t]
\centering
\includegraphics[width = 3.5in]{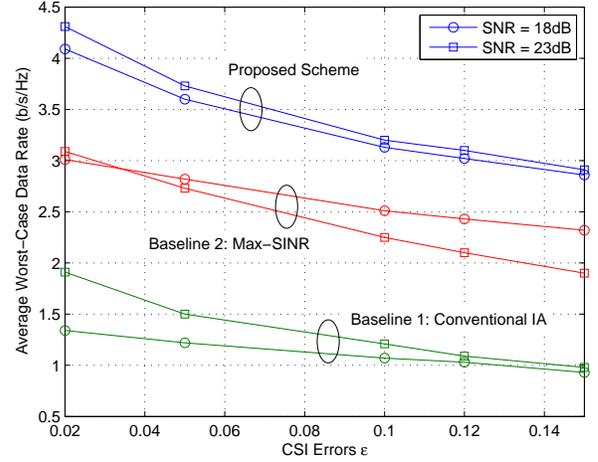}
\caption{Average worst-case data rate versus CSI errors. $K=3$,
$N=M=4$, $L=2$ and SNR 18dB and 23dB.}\label{Fig:Imperfect}
\end{figure}

\section{Conclusions} \label{Sec:Conclusions}
In this paper, we proposed a robust transceiver design for the
$K$-pair quasi-static MIMO interference channel with fairness
considerations. Specifically, we formulated the
precoder-decorrelator design as an optimization problem to maximize
the worst-case SINR among all users. We devised a low complexity
iterative algorithm based on AO and SDR techniques. Numerical
results verify the advantages of incorporating into transceiver
design for the interference channel important practical issues such
as CSI uncertainty and fairness performance.

\appendices
\begin{figure*}[!t]
\normalsize
\begin{IEEEeqnarray*}{l}
\gamma_l^{(k)} ( \mathcal{H}, \{ \{ \textbf{v}_m^{(j)}
\}_{m=1}^{L_j} \}_{j=1}^K, \textbf{u}_l^{(k)} )\\
\textstyle= \frac{|| ( \textbf{u}_l^{(k)} )^\dag (
\widehat{\textbf{H}}^{(k,k)}\!+ \mathbf{\Delta}^{(k,k)} )
\textbf{v}_l^{(k)} ||^2}{\sum_{\substack{m=1\\m \neq l}}^{L_k} || (
\textbf{u}_l^{(k)} )^\dag ( \widehat{\textbf{H}}^{(k,k)}\!+
\mathbf{\Delta}^{(k,k)} ) \textbf{v}_m^{(k)} ||^2 +
\sum_{\substack{j=1\\j \neq k}}^K \sum_{m=1}^{L_j} || (
\textbf{u}_l^{(k)} )^\dag ( \widehat{\textbf{H}}^{(k,j)}\!+
\mathbf{\Delta}^{(k,j)} ) \textbf{v}_m^{(j)} ||^2 + N_0 ||
\textbf{u}_l^{(k)} ||^2}\IEEEyessubnumber\label{Eqn:Sinr_perfect0}\\
\textstyle \ge \frac{|| ( \textbf{u}_l^{(k)} )^\dag
\widehat{\textbf{H}}^{(k,k)} \textbf{v}_l^{(k)} ||^2 - || (
\textbf{u}_l^{(k)} )^\dag \mathbf{\Delta}^{(k,k)} \textbf{v}_l^{(k)}
||^2}{\sum_{j=1}^K \sum_{m=1}^{L_j} ( || ( \textbf{u}_l^{(k)} )^\dag
\widehat{\textbf{H}}^{(k,j)} \textbf{v}_m^{(j)} ||^2 + || (
\textbf{u}_l^{(k)} )^\dag \mathbf{\Delta}^{(k,j)} \textbf{v}_m^{(j)}
||^2 ) - (|| ( \textbf{u}_l^{(k)} )^\dag
\widehat{\textbf{H}}^{(k,k)} \textbf{v}_l^{(k)} ||^2 + || (
\textbf{u}_l^{(k)} )^\dag \mathbf{\Delta}^{(k,k)} \textbf{v}_l^{(k)}
||^2 ) + N_0 ||
\textbf{u}_l^{(k)} ||^2}\IEEEyessubnumber\label{Eqn:Sinr_perfect1}\\
\textstyle \ge \frac{|| ( \textbf{u}_l^{(k)} )^\dag
\widehat{\textbf{H}}^{(k,k)} \textbf{v}_l^{(k)} ||^2 - \varepsilon
|| \textbf{u}_l^{(k)} ||^2 || \textbf{v}_l^{(k)} ||^2}{\sum_{j=1}^K
\sum_{m=1}^{L_j} || ( \textbf{u}_l^{(k)} )^\dag
\widehat{\textbf{H}}^{(k,j)} \textbf{v}_m^{(j)} ||^2 + \varepsilon
|| \textbf{u}_l^{(k)} ||^2 \sum_{j=1}^K \sum_{m=1}^{L_j} ||
\textbf{v}_m^{(j)} ||^2 - || ( \textbf{u}_l^{(k)} )^\dag
\widehat{\textbf{H}}^{(k,k)} \textbf{v}_l^{(k)} ||^2 - \varepsilon
|| \textbf{u}_l^{(k)} ||^2 || \textbf{v}_l^{(k)} ||^2 + N_0 ||
\textbf{u}_l^{(k)} ||^2}\IEEEyessubnumber\label{Eqn:Sinr_perfect2}\\
\triangleq \widetilde{\gamma}_l^{(k)} ( \widehat{\mathcal{H}}, \{ \{
\textbf{v}_m^{(j)} \}_{m=1}^{L_j} \}_{j=1}^K, \textbf{u}_l^{(k)} ).
\end{IEEEeqnarray*}
\hrulefill
\end{figure*}

\section{Proof: Worst-Case SINR with Imperfect
CSIT}\label{Proof:Worst-case_SINR} Given CSI estimates
$\widehat{\mathcal{H}} = \{ \widehat{\textbf{H}}^{(k,j)}
\}_{j,k=1}^K$ at the transmitter, the worst-case SINR for each data
stream estimate can be expressed as follows. Consider
$\widetilde{s}_l^{(k)}$ whose SINR $\gamma_l^{(k)} ( \mathcal{H}, \{
\{ \textbf{v}_m^{(j)} \}_{m=1}^{L_j} \}_{j=1}^K, \textbf{u}_l^{(k)}
)$ is given by (\ref{Eqn:Sinr_perfect0}). First, by the triangle
inequality
\begin{equation*}
\begin{array}{l}
|| ( \textbf{u}_l^{(k)} )^\dag (
\widehat{\textbf{H}}^{(k,j)}\!\!+\!\mathbf{\Delta}^{(k,j)} )
\textbf{v}_m^{(j)} ||^2\\
\;\;\;\;\;\;\;\;\;\;\ge || ( \textbf{u}_l^{(k)} )^\dag
\widehat{\textbf{H}}^{(k,j)} \textbf{v}_m^{(j)} ||^2 - || (
\textbf{u}_l^{(k)} )^\dag \mathbf{\Delta}^{(k,j)} \textbf{v}_m^{(j)}
||^2,\\
|| ( \textbf{u}_l^{(k)} )^\dag (
\widehat{\textbf{H}}^{(k,j)}\!\!+\!\mathbf{\Delta}^{(k,j)} )
\textbf{v}_m^{(j)} ||^2\\
\;\;\;\;\;\;\;\;\;\;\le || ( \textbf{u}_l^{(k)} )^\dag
\widehat{\textbf{H}}^{(k,j)} \textbf{v}_m^{(j)} ||^2 + || (
\textbf{u}_l^{(k)} )^\dag \mathbf{\Delta}^{(k,j)} \textbf{v}_m^{(j)}
||^2,
\end{array}
\end{equation*}
and so the SINR is lowered bounded as (\ref{Eqn:Sinr_perfect1}).
Second, with CSI error $|| \mathbf{\Delta}^{(k,j)} ||^2 \leq
\varepsilon$,
\begin{equation*}
\begin{array}{l}
|| ( \textbf{u}_l^{(k)} )^\dag \mathbf{\Delta}^{(k,j)}
\textbf{v}_m^{(j)} ||^2\\
= \textrm{Tr} ( ( \textbf{u}_l^{(k)} )^\dag
\mathbf{\Delta}^{(k,j)} \textbf{v}_m^{(j)} ( \textbf{v}_m^{(j)}
)^\dag ( \mathbf{\Delta}^{(k,j)} )^\dag
\textbf{u}_l^{(k)} )\\
\stackrel{(a)}{\leq} \textrm{Tr} ( \textbf{u}_l^{(k)} (
\textbf{u}_l^{(k)} )^\dag ) \textrm{Tr} ( \mathbf{\Delta}^{(k,j)}
\textbf{v}_m^{(j)} (
\textbf{v}_m^{(j)} )^\dag ( \mathbf{\Delta}^{(k,j)} )^\dag )\\
\stackrel{(b)}{\leq} \textrm{Tr} ( \textbf{u}_l^{(k)} (
\textbf{u}_l^{(k)} )^\dag ) \underbrace{\textrm{Tr} ( (
\mathbf{\Delta}^{(k,j)} )^\dag \mathbf{\Delta}^{(k,j)} )}_{= ||
\mathbf{\Delta}^{(k,j)} ||^2} \textrm{Tr} ( \textbf{v}_m^{(j)} (
\textbf{v}_m^{(j)} )^\dag )\\
= \varepsilon || \textbf{u}_l^{(k)} ||^2 || \textbf{v}_m^{(j)} ||^2,
\end{array}
\end{equation*}
where (a) and (b) follow from the properties that $\textrm{Tr} (
\textbf{A} \textbf{B} ) = \textrm{Tr} ( \textbf{B} \textbf{A} )$ for
$\textbf{A} \in \mathbb{C}^{M \times N}$ and $\textbf{B} \in
\mathbb{C}^{N \times M}$ and $\textrm{Tr} ( \textbf{C} \textbf{D} )
\leq \textrm{Tr} ( \textbf{C} ) \textrm{Tr} ( \textbf{D} )$ for
positive semi-definite $\textbf{C}, \textbf{D} \in \mathbb{C}^{N
\times N}$. Thus, the worst-case SINR perceived by the transmitter
can be expressed as (\ref{Eqn:Sinr_perfect2}).

\section{Proof: Optimal Decorrelator with Fixed Precoders}
\label{Proof:Optimal_decorrelator} From (\ref{Eqn:Worst-case_SINR}),
the worst-case SINR of data stream estimate $\widetilde{s}_l^{(k)}$
can be expressed as
\begin{equation}
\textstyle\widetilde{\gamma}_l^{(k)} ( \widehat{\mathcal{H}}, \{ \{
\textbf{v}_m^{(j)} \}_{m=1}^{L_j} \}_{j=1}^K, \textbf{u}_l^{(k)} ) =
\frac{( \textbf{u}_l^{(k)} )^\dag \textbf{E}_l^{(k)}
\textbf{u}_l^{(k)}}{( \textbf{u}_l^{(k)} )^\dag \textbf{F}_l^{(k)}
\textbf{u}_l^{(k)}},\label{Eqn:Sinr_quotient}
\end{equation}
where
\begin{IEEEeqnarray*}{l}
\textbf{F}_l^{(k)} \textstyle= \sum_{j=1}^K \sum_{m=1}^{L_j}
\widehat{\textbf{H}}^{(k,j)} \textbf{v}_m^{(j)} ( \textbf{v}_m^{(j)}
)^\dag ( \widehat{\textbf{H}}^{(k,j)} )^\dag\\
\textstyle+ \varepsilon \sum_{j=1}^K \sum_{m=1}^{L_j} ||
\textbf{v}_m^{(j)} ||^2 \textbf{I}_N - \widehat{\textbf{H}}^{(k,k)}
\textbf{v}_l^{(k)} ( \textbf{v}_l^{(k)} )^\dag (
\widehat{\textbf{H}}^{(k,k)} )^\dag\\
- \varepsilon || \textbf{v}_l^{(k)} ||^2 \textbf{I}_N + N_0
\textbf{I}_N,
\end{IEEEeqnarray*}
which is a Hermitian and positive definite matrix, and
\begin{equation*}
\begin{array}{l}
\textbf{E}_l^{(k)} = \widehat{\textbf{H}}^{(k,k)} \textbf{v}_l^{(k)}
( \textbf{v}_l^{(k)} )^\dag ( \widehat{\textbf{H}}^{(k,k)} )^\dag -
\varepsilon || \textbf{v}_l^{(k)} ||^2 \textbf{I}_N,
\end{array}\label{Eqn:Sinr_num}
\end{equation*}
which is a non-negative definite\footnote{If $\textbf{E}_l^{(k)}$ is
negative definite, then the CSI error $\varepsilon$ is too high.
Without loss of generality, we assume $\varepsilon$ is sufficiently
small.} Hermitian matrix. Without loss of generality, let
$\textbf{u}_l^{(k)} = c ( \textbf{F}_l^{(k)} )^{-\frac{1}{2}}
\textbf{w}_l^{(k)}$ for arbitrary scaling factor $c \in \mathbb{C}$.
We can equivalently expressed (\ref{Eqn:Sinr_quotient}) as
\begin{IEEEeqnarray*}{Rl}
\textstyle\frac{( \textbf{u}_l^{(k)} )^\dag \textbf{E}_l^{(k)}
\textbf{u}_l^{(k)}}{( \textbf{u}_l^{(k)} )^\dag \textbf{F}_l^{(k)}
\textbf{u}_l^{(k)}} &\textstyle= \frac{( \textbf{w}_l^{(k)} )^\dag (
\textbf{F}_l^{(k)} )^{-\frac{1}{2}} \textbf{E}_l^{(k)} (
\textbf{F}_l^{(k)} )^{-\frac{1}{2}} \textbf{w}_l^{(k)}}{(
\textbf{w}_l^{(k)} )^\dag
\textbf{w}_l^{(k)}}\IEEEyesnumber\label{Eqn:Sinr_max}\\
&\textstyle= \frac{( \textbf{w}_l^{(k)} )^\dag \textbf{Q}
\mathbf{\Lambda} \textbf{Q}^\dag \textbf{w}_l^{(k)}}{(
\textbf{w}_l^{(k)} )^\dag \textbf{w}_l^{(k)}},
\end{IEEEeqnarray*}
where $\textbf{Q} \mathbf{\Lambda} \textbf{Q}^\dag$ denotes the
eigen-decomposition of $( \textbf{F}_l^{(k)} )^{-\frac{1}{2}}
\textbf{E}_l^{(k)} ( \textbf{F}_l^{(k)} )^{-\frac{1}{2}}$. It can be
shown that\footnote{Please refer to
\cite[Appendix~E]{Bok:Adaptive_filter_theory:Haykin}.}
(\ref{Eqn:Sinr_max}) is maximized with $( \textbf{w}_l^{(k)}
)^\star$ being the principal eigenvector\footnote{As per
\cite[Theorem~7.6.3]{Bok:Matrix_analysis:Horn} the principle
eigenvalue of $( \textbf{F}_l^{(k)} )^{-\frac{1}{2}}
\textbf{E}_l^{(k)} ( \textbf{F}_l^{(k)} )^{-\frac{1}{2}}$ is always
positive.} of $( \textbf{F}_l^{(k)} )^{-\frac{1}{2}}
\textbf{E}_l^{(k)} ( \textbf{F}_l^{(k)} )^{-\frac{1}{2}}$. In turn,
the optimal unit norm decorrelator is given by $( \textbf{u}_l^{(k)}
)^\star = \frac{( \textbf{F}_l^{(k)} )^{-\frac{1}{2}} (
\textbf{w}_l^{(k)} )^\star}{|| ( \textbf{F}_l^{(k)} )^{-\frac{1}{2}}
( \textbf{w}_l^{(k)} )^\star ||}$.

\section{Proof: Optimality of the SDR Solution for
Problem~$\widetilde{\mathcal{Q}}_{\textrm{v}}$}
\label{Proof:Optimiality_conditions} By using SDR, we solve the
following SDP problem with complex-valued parameters:
\begin{IEEEeqnarray*}{cl}\label{Eqn:HH_primal}
\min_{\textbf{V}_m^{(j)}, \Xi}&\;\; \Xi\IEEEyessubnumber\label{Eqn:R_V_cost}\\
\textrm{s. t.} & \textstyle\sum\nolimits_{m=1}^{L_j} \textrm{Tr} (
\textbf{V}_m^{(j)} ) \leq \rho_j \Xi,\forall j \in \mathcal{K}\\
&\textstyle\sum_{j=1}^K\!\sum_{m=1}^{L_j}\!\text{Tr}(\mathbf{A}_{(l,m)}^{(k,j)}\mathbf{V}_{m}^{(j)})\!\geq\!
b_{l}^{(k)}\!, \forall l\!\in\!\mathcal{L}_k, \forall k\!
\in\!\mathcal{K},\;\;\;\;\;\;\;\IEEEyessubnumber\label{Eqn:R_V_QoS}\\
&\Xi\geq 0,\IEEEyessubnumber\\
&\textbf{V}_m^{(j)} \succeq 0, \forall m \in \mathcal{L}_j, \forall
j \in \mathcal{K},\IEEEyessubnumber\label{Eqn:R_V_SD}
\end{IEEEeqnarray*}
where $\mathbf{A}_{(l,m)}^{(k,j)}\in\mathbb{H}^{M}$ is given by
\begin{equation*}
\mathbf{A}_{(l,m)}^{(k,j)}\!\!=\!\!\left\{\!\!\!\!
\begin{array}{ll}
( \widehat{\textbf{H}}^{(k,k)}\!)^\dag \textbf{u}_l^{(k)}\!(
\textbf{u}_l^{(k)} )^\dag
\widehat{\textbf{H}}^{(k,k)}\!\!-\!\varepsilon || \textbf{u}_l^{(k)}
||^2\mathbf{I}\!\!\!\!\!\!& \begin{array}{l}\text{if }
j\!=\!k\\\text{and }
m\!=\!l\end{array}\\
-\!\gamma(( \widehat{\textbf{H}}^{(k,j)}\!)^\dag
\textbf{u}_l^{(k)}\!( \textbf{u}_l^{(k)} )^\dag
\widehat{\textbf{H}}^{(k,j)}\!\!\!+\!\varepsilon ||
\textbf{u}_l^{(k)} ||^2\mathbf{I} )\!\!\!\!\!\!& \text{ otherwise }
\end{array}\right.\;\;\;\;\;\;
\end{equation*}
\addtocounter{equation}{2}
\begin{figure*}[!t]
\normalsize
\begin{IEEEeqnarray*}{ll}
\mathbf{Z}_{m}^{(j)}&\textstyle=x^{(j)}\mathbf{I}-\sum_{k=1}^K\sum_{l=1}^{L_k}y_{l}^{(k)}
\rho_j\mathbf{A}_{(l,m)}^{(k,j)}\IEEEyesnumber\label{Eqn:wxy}\\
&=\underbrace{x^{(j)}\mathbf{I}+\textstyle\sum_{k=1}^K\textstyle\sum_{l=1}^{L_k}y_{l}^{(k)}\rho_j\gamma\left((
\widehat{\textbf{H}}^{(k,j)} )^\dag \textbf{u}_l^{(k)} (
\textbf{u}_l^{(k)} )^\dag \widehat{\textbf{H}}^{(k,j)} + \varepsilon
|| \textbf{u}_l^{(k)} ||^2\mathbf{I} \right)\mathcal{I}{\{k\neq
j\&l\neq m\}}+y_{m}^{(j)}\rho_j\varepsilon ||
\textbf{u}_m^{(j)}||^2\mathbf{I}}_{\text{rank } M} \\
&-\underbrace{y_{m}^{(j)}\rho_j( \widehat{\textbf{H}}^{(j,j)} )^\dag
\textbf{u}_m^{(j)} ( \textbf{u}_m^{(j)} )^\dag
\widehat{\textbf{H}}^{(j,j)}}_{\text{rank } 1}.
\end{IEEEeqnarray*}
\hrulefill
\end{figure*}
\addtocounter{equation}{-3}and $b_{l}^{(k)}=\gamma N_0 ||
\textbf{u}_l^{(k)} ||^2>0$. The corresponding dual problem is given
by the following SDP:
\begin{IEEEeqnarray*}{cl}\label{Eqn:HH_dual}
\max_{\substack{y_{l}^{(k)}\\x^{(j)}}}&\;
\textstyle\sum\nolimits_{k=1}^K\textstyle\sum\nolimits_{l=1}^{L_k}y_{l}^{(k)}b_{l}^{(k)}\IEEEyessubnumber\label{Eqn:ggg}\\
\textrm{s. t.}
&\;\underbrace{x^{(j)}\mathbf{I}\!-\!\!\textstyle\sum_{k=1}^K\!\sum_{l=1}^{L_k}y_{l}^{(k)}
\!\!\rho_j\mathbf{A}_{(l,m)}^{(k,j)}}_{=\mathbf{Z}_{m}^{(j)}}\!\succeq\!0,
\forall m\!\!\in\!\mathcal{L}_j, \forall j\!\!\in\!\mathcal{K},\;\;\;\;\;\IEEEyessubnumber\label{Eqn:xxx}\\
&1-\textstyle\sum_{j=1}^Kx^{(j)} \geq 0,\IEEEyessubnumber\\
&y_{l}^{(k)}\geq 0, \forall k \in \mathcal{K}, \forall l \in
\mathcal{L}_k,\IEEEyessubnumber\\
&x^{(j)}\geq 0, \forall j \in
\mathcal{K},\IEEEyessubnumber\label{Eqn:nnn}
\end{IEEEeqnarray*}

Note that $(\mathbf{V}_{m}^{(j)})^\star \neq\mathbf{0},\forall j \in
\mathcal{K}, \forall m \in \mathcal{L}_j$, and from the
complementary conditions for the primal and dual SDP:
\begin{equation}
\label{Eqn:HH_comple}
\text{Tr}(\mathbf{Z}_{m}^{(j)}(\mathbf{V}_{m}^{(j)})^\star)=0,\forall
j \in \mathcal{K}, \forall m \in \mathcal{L}_j
\end{equation}
\addtocounter{equation}{1} we can infer that
$\mathbf{Z}_{m}^{(j)}\not\succ\mathbf{0}$. Suppose that one of the
optimal values $\{\{(y_{l}^{(k)})^\star \}_{l=1}^{L_k}\}_{k=1}^K$
for the dual problem, say $(y_{1}^{(1)})^\star=0$, then
\begin{equation*}
\mathbf{Z}_{1}^{(1)}=x^{(1)}\mathbf{I}+\sum_{k=1}^K
\sum_{l=1}^{L_k}y_{l}^{(k)} (-
\rho_1\mathbf{A}_{(l,1)}^{(k,1)})\mathcal{I}{\{k\neq 1\&l\neq
1\}}\succ\mathbf{0}.
\end{equation*}
It contradicts the fact $\mathbf{Z}_{1}^{(1)}\not\succ\mathbf{0}$,
and hence $(y_{l}^{(k)})^\star>0, \forall k \in \mathcal{K}, \forall
l \in \mathcal{L}_k$. From (\ref{Eqn:xxx}) and (\ref{Eqn:wxy}),
$\text{rank}(\mathbf{Z}_{m}^{(j)}) \ge M-1$. On the other hand, from
(\ref{Eqn:HH_comple}), since $\textbf{Z}_m^{(j)}\nsucc \textbf{0}$
so $\text{rank}(\mathbf{Z}_{m}^{(j)}) < M$. It follows that
$\text{rank}(\mathbf{Z}_{m}^{(j)})=M-1$. Moreover, due to
(\ref{Eqn:HH_comple}) the optimal solution
$\{\{(\mathbf{V}_{m}^{(j)})^\star\}_{m=1}^{L_j}\}_{j=1}^K$ of primal
problem (\ref{Eqn:HH_primal}) must be of rank one. In other words,
there will be zero duality gap between the primal non-convex problem
$\widetilde{\mathcal{Q}}_\textrm{v}$ and the dual problem obtained
by relaxing the rank constraint given by (\ref{Eqn:HH_dual}).

\section{Proof: Convergence of
Algorithm~\ref{Algorithm:Top-level}} \label{Proof:Convergence} At
the $n^{\textrm{th}}$ iteration of
Algorithm~\ref{Algorithm:Top-level}, we denote the precoders as $\{
\{ \widetilde{\textbf{v}}_m^{(j)}[n] \}_{m=1}^{L_j} \}_{j=1}^K$, the
decorrelators as $\{ \{ \widetilde{\textbf{u}}_m^{(j)}[n]
\}_{m=1}^{L_j} \}_{j=1}^K$, the minimum SINR as
$\widetilde{\gamma}[n]$, and the transmit power scaling factor as
$\widetilde{\beta}[n]$.

Upon initialization, we \emph{define} the minimum SINR as
$\widetilde{\gamma}[0] = 0$ and start with arbitrary precoders $\{
\{ \widetilde{\textbf{v}}_m^{(j)}[0] \}_{m=1}^{L_j} \}_{j=1}^K$,
where the transmit power of the $j^{\textrm{th}}$ source node is
$\sum_{m=1}^{L_j} ( \widetilde{\textbf{v}}_m^{(j)}[0] )^\dag
\widetilde{\textbf{v}}_m^{(j)}[0] = P_j$, and the transmit power
scaling factor is $\widetilde{\beta}[0] = \min( P_1, \ldots, P_K )$.

In the following, we show that each iteration of
Algorithm~\ref{Algorithm:Top-level} increases the minimum SINR, i.e.
$\widetilde{\gamma}[n] \ge \widetilde{\gamma}[n-1]$, so
Algorithm~\ref{Algorithm:Top-level} must converge.

In Step 1, given the precoders $\{ \{
\widetilde{\textbf{v}}_m^{(j)}[n-1] \}_{m=1}^{L_j} \}_{j=1}^K$, the
decorrelators $\{ \{ \widetilde{\textbf{u}}_m^{(j)}[n]
\}_{m=1}^{L_j} \}_{j=1}^K$ are optimized to increase the minimum
SINR, i.e.
\begin{IEEEeqnarray*}{Rl}
\widehat{\gamma} &= \displaystyle \min_{\substack{l\in
\mathcal{L}_k\\k\in \mathcal{K}}} \textstyle
\widetilde{\gamma}_l^{(k)} ( \widehat{\mathcal{H}}, \{ \{
\widetilde{\textbf{v}}_m^{(j)}[n\!-\!1]
\}_{m=1}^{L_j} \}_{j=1}^K, \widetilde{\textbf{u}}_l^{(k)}[n] )\\
&\ge \displaystyle \min_{\substack{l\in \mathcal{L}_k\\k\in
\mathcal{K}}} \textstyle \widetilde{\gamma}_l^{(k)} (
\widehat{\mathcal{H}}, \{ \{ \widetilde{\textbf{v}}_m^{(j)}[n\!-\!1]
\}_{m=1}^{L_j} \}_{j=1}^K,
\widetilde{\textbf{u}}_l^{(k)}[n\!-\!1] )\\
&=\widetilde{\gamma}[n\!-\!1].\IEEEyesnumber\label{Eqn:Convergence1}
\end{IEEEeqnarray*}
In Step 3, given the decorrelators $\{ \{
\widetilde{\textbf{u}}_m^{(j)}[n] \}_{m=1}^{L_j} \}_{j=1}^K$ and the
minimum SINR constraint $\widehat{\gamma}$, the precoders $\{ \{
\widetilde{\textbf{v}}_m^{(j)}[n] \}_{m=1}^{L_j} \}_{j=1}^K$ are
optimized to jointly reduce the transmit powers of all nodes, i.e.
the minimum SINR is unchanged
\begin{IEEEeqnarray*}{Rl}
\widehat{\gamma} &= \displaystyle \min_{\substack{l\in
\mathcal{L}_k\\k\in \mathcal{K}}} \textstyle
\widetilde{\gamma}_l^{(k)} ( \widehat{\mathcal{H}}, \{ \{
\widetilde{\textbf{v}}_m^{(j)}[n] \}_{m=1}^{L_j}
\}_{j=1}^K, \widetilde{\textbf{u}}_l^{(k)}[n] )\\
&= \displaystyle \min_{\substack{l\in \mathcal{L}_k\\k\in
\mathcal{K}}} \textstyle \widetilde{\gamma}_l^{(k)} (
\widehat{\mathcal{H}}, \{ \{ \widetilde{\textbf{v}}_m^{(j)}[n\!-\!1]
\}_{m=1}^{L_j} \}_{j=1}^K, \widetilde{\textbf{u}}_l^{(k)}[n] )
\end{IEEEeqnarray*}
whereas the transmit powers of all source nodes are reduced
\begin{IEEEeqnarray*}{Rl}
\rho_j \widetilde{\beta}[n] &= \textstyle\sum_{m=1}^{L_j} (
\widetilde{\textbf{v}}_m^{(j)}[n] )^\dag
\widetilde{\textbf{v}}_m^{(j)}[n]\\
&\leq \textstyle\sum_{m=1}^{L_j} (
\widetilde{\textbf{v}}_m^{(j)}[n-1] )^\dag
\widetilde{\textbf{v}}_m^{(j)}[n-1]\\
&= P_j.
\end{IEEEeqnarray*}
In Step 5, the precoders are up-scaled to the power constraint,
i.e.$\textbf{v}_m^{(j)}[n] = \sqrt{P_j/(\rho_j
\widetilde{\beta}[n])} \textbf{v}_m^{(j)}[n]$, where by definition
$P_1/\rho_1 = \ldots = P_K/\rho_K$. As such, the minimum SINR is
increased according to
\begin{IEEEeqnarray*}{l}
\widehat{\gamma} = \displaystyle \min_{\substack{l\in
\mathcal{L}_k\\k\in \mathcal{K}}} \textstyle
\widetilde{\gamma}_l^{(k)} ( \widehat{\mathcal{H}}, \{ \{
\textbf{v}_m^{(j)}[n] \}_{m=1}^{L_j} \}_{j=1}^K,
\textbf{u}_l^{(k)}[n] )\\
=\!\displaystyle \min_{\substack{l\in \mathcal{L}_k\\k\in
\mathcal{K}}}\!\textstyle \frac{|| ( \textbf{u}_l^{(k)}[n] )^\dag
\widehat{\textbf{H}}^{(k,k)} \textbf{v}_l^{(k)}[n] ||^2 -
\varepsilon || \textbf{u}_l^{(k)}[n] ||^2 || \textbf{v}_l^{(k)}[n]
||^2}{\left(\substack{\sum_{j=1}^K\!\sum_{m=1}^{L_j} || (
\textbf{u}_l^{(k)}[n] )^\dag \widehat{\textbf{H}}^{(k,j)}
\textbf{v}_m^{(j)}[n] ||^2\\\!+ \varepsilon || \textbf{u}_l^{(k)}[n]
||^2 \sum_{j=1}^K \sum_{m=1}^{L_j} || \textbf{v}_m^{(j)}[n] ||^2\!-
\varepsilon || \textbf{u}_l^{(k)}[n] ||^2 || \textbf{v}_l^{(k)}[n]
||^2\\- || ( \textbf{u}_l^{(k)}[n] )^\dag
\widehat{\textbf{H}}^{(k,k)} \textbf{v}_l^{(k)}[n] ||^2 + N_0 ||
\textbf{u}_l^{(k)}[n]
||^2}\right)}\\
< \displaystyle \min_{\substack{l\in \mathcal{L}_k\\k\in
\mathcal{K}}} \textstyle \widetilde{\gamma}_l^{(k)} (
\widehat{\mathcal{H}}, \{ \{ \sqrt{P_K/(\rho_K\widetilde{\beta}[n])}
\textbf{v}_m^{(j)}[n]
\}_{m=1}^L \}_{j=1}^K, \textbf{u}_l^{(k)}[n] )\\
=\!\displaystyle \min_{\substack{l\in \mathcal{L}_k\\k\in
\mathcal{K}}}\!\textstyle \frac{|| ( \textbf{u}_l^{(k)}[n] )^\dag
\widehat{\textbf{H}}^{(k,k)} \textbf{v}_l^{(k)}[n] ||^2 -
\varepsilon || \textbf{u}_l^{(k)}[n] ||^2 || \textbf{v}_l^{(k)}[n]
||^2}{\left(\substack{\sum_{j=1}^K\!\sum_{m=1}^{L_j} || (
\textbf{u}_l^{(k)}[n] )^\dag \widehat{\textbf{H}}^{(k,j)}
\textbf{v}_m^{(j)}[n] ||^2\\\!+ \varepsilon || \textbf{u}_l^{(k)}[n]
||^2 \sum_{j=1}^K \sum_{m=1}^{L_j} || \textbf{v}_m^{(j)}[n] ||^2\!-
\varepsilon || \textbf{u}_l^{(k)}[n] ||^2 || \textbf{v}_l^{(k)}[n]
||^2\\- || ( \textbf{u}_l^{(k)}[n] )^\dag
\widehat{\textbf{H}}^{(k,k)} \textbf{v}_l^{(k)}[n] ||^2 +
\frac{N_0}{(P_K/\rho_K)(1/\widetilde{\beta}[n])} ||
\textbf{u}_l^{(k)}[n] ||^2}\right)}\\
= \widetilde{\gamma}[n].\IEEEyesnumber\label{Eqn:Convergence2}
\end{IEEEeqnarray*}
It follows from (\ref{Eqn:Convergence1}) and
(\ref{Eqn:Convergence2}) that the minimum SINR increases with each
iteration, i.e. $\widetilde{\gamma}[n] \ge \widehat{\gamma} \ge
\widetilde{\gamma}[n\!-\!1]$, and
Algorithm~\ref{Algorithm:Top-level} must converge.

\bibliographystyle{IEEEtran}
\bibliography{IEEEabrv,myBibFile}

\begin{IEEEbiography}[{\includegraphics[width=1in,height=1.25in,clip,keepaspectratio]{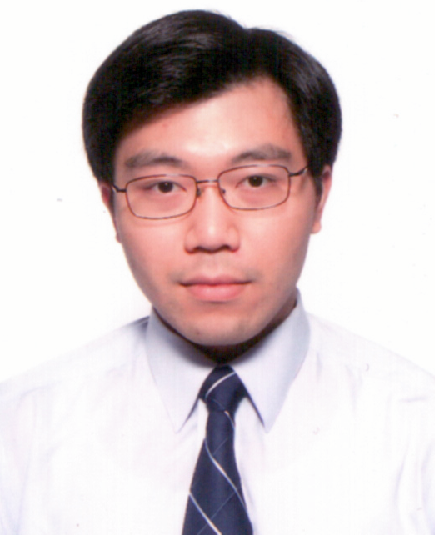}}]{Eddy
Chiu} received the B.A.Sc. (Honors) and M.A.Sc. degrees from Simon
Fraser University, Canada, in 2003 and 2006, respectively, both in
Electrical Engineering. Currently, he is working towards the Ph.D.
degree at the Department of Electronic and Computer Engineering,
Hong Kong University of Science and Technology. His research
interests include MIMO communications with limited feedback,
relay-assisted communications, and interference mitigation
techniques.
\end{IEEEbiography}

\vfill

\begin{IEEEbiography}[{\includegraphics[width=1in,height=1.25in,clip,keepaspectratio]{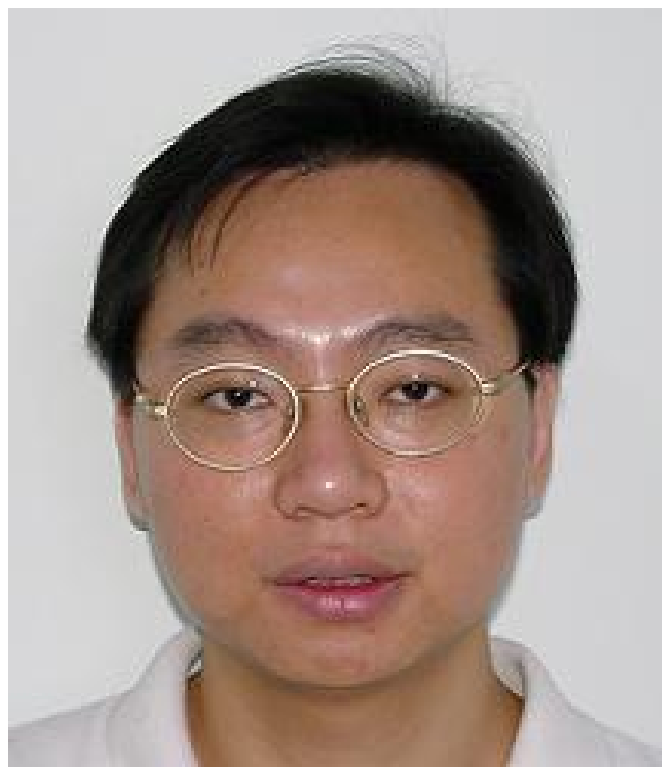}}]{Vincent
K. N. Lau} obtained B.Eng (Distinction 1st Hons) from the University
of Hong Kong in 1992 and Ph.D. from Cambridge University in 1997. He
was with PCCW as system engineer from 1992-1995 and Bell Labs -
Lucent Technologies as member of technical staff from 1997-2003. He
then joined the Department of Electronic and Computer Engineering,
Hong Kong University of Science and Technology as Professor. His
current research interests include robust and delay-sensitive
cross-layer scheduling of MIMO/OFDM wireless systems with imperfect
channel state information, cooperative and cognitive communications
as well as stochastic approximation and Markov Decision Process.
\end{IEEEbiography}

\vfill

\begin{IEEEbiography}[{\includegraphics[width=1in,height=1.25in,clip,keepaspectratio]{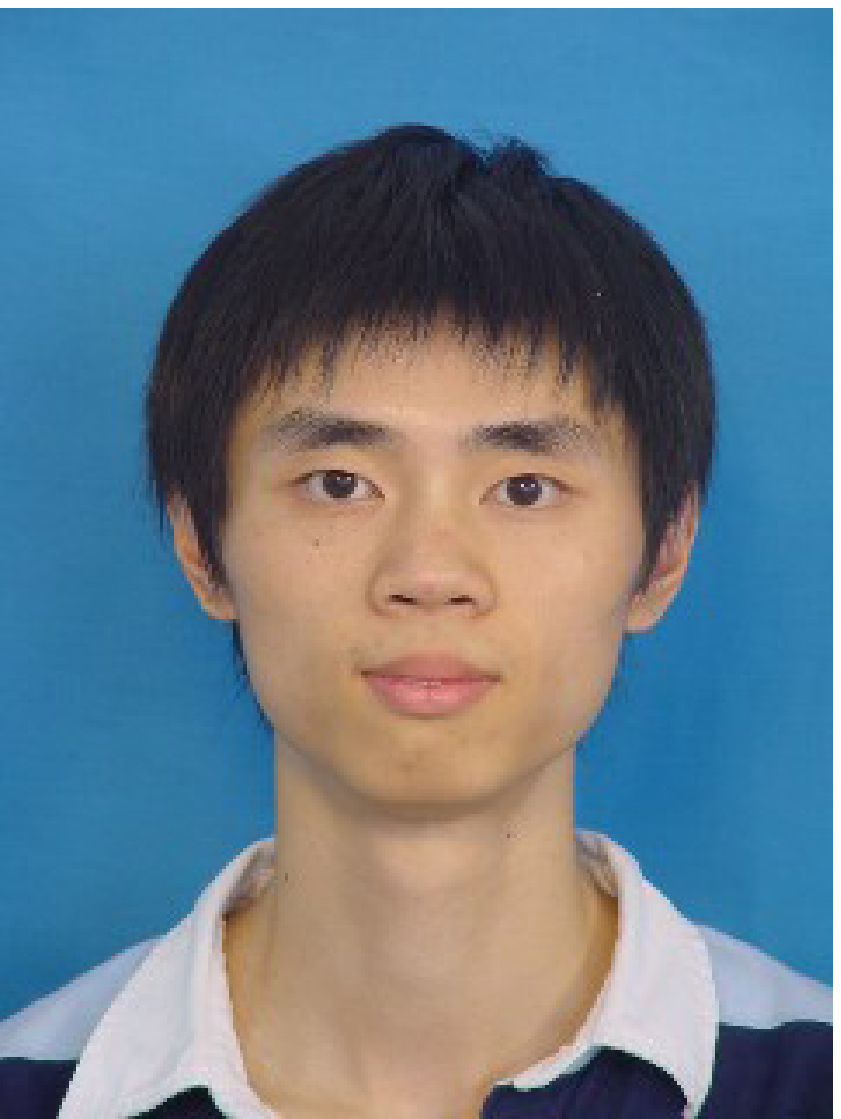}}]{Huang
Huang} received the B.Eng. and M.Eng. (Gold medal) from the Harbin
Institute of Technology (HIT) in 2005 and 2007 respectively, all in
Electrical Engineering. He is currently a PhD student at the
Department of Electronic and Computer Engineering, Hong Kong
University of Science and Technology. His recent research interests
include cross layer design, interference management in interference
network, and embedded system design.
\end{IEEEbiography}

\vfill

\begin{IEEEbiographynophoto}{Tao Wu}
\end{IEEEbiographynophoto}

\vfill

\begin{IEEEbiographynophoto}{Sheng Liu}
\end{IEEEbiographynophoto}

\vfill

\end{document}